\documentclass[sts]{imsart}
\usepackage{geometry}
\usepackage{subfiles}
\usepackage{url}
\usepackage{lineno}
\usepackage{amsfonts}
\usepackage{amssymb}
\usepackage{mathrsfs}
\usepackage{textgreek}
\usepackage{mathtools}
\usepackage{amsthm,amsmath}
\usepackage{bbm}
\usepackage{bm}
\usepackage{tikz}
\usepackage[colorlinks=true,
            linkcolor=blue,
            urlcolor=blue,
            citecolor=blue,pdftex,pagebackref=true]{hyperref}
\usepackage[nameinlink]{cleveref}
\usepackage[round]{natbib}
\usepackage{enumitem}
\usepackage{fancyhdr}
\usepackage{float}
\usepackage{soul}
\usepackage{xr}

\startlocaldefs
\endlocaldefs
\RequirePackage[OT1]{fontenc}
\usetikzlibrary{positioning, shapes.geometric, arrows, shapes, snakes, automata, backgrounds, petri}
\usetikzlibrary{shapes.geometric}

\def\BE{\mathbb{E}}
\def\BP{\mathbb{P}}

\def\IIFF{\mathbb{IF}}

\def\IF{\mathsf{IF}}
\def\log{\mathsf{log}}

\def\var{\mathsf{var}}

\def\cov{\mathsf{cov}}

\def\Bias{\mathsf{Bias}}
\def\H{\mathsf{H}}

\def\se{\mathsf{s.e.}}

\def\2{$\text{2}^{\text{nd}}$-order influence function $U$-statistics}
\def\3{$\text{3}^{\text{rd}}$-order influence function $U$-statistics}
\def\4{$\text{4}^{\th}$-order influence function $U$-statistics}
\def\5{$\text{5}^{\th}$-order influence function $U$-statistics}

\def\th{\mathsf{th}}

\def\EB{\mathsf{EB}}
\def\TB{\mathsf{TB}}
\def\CSBias{\mathsf{CSBias}}

\def\P{\mathsf{P}}

\def\zbar{\bar{\mathsf{z}}}

\def\Holder{\text{H\"{o}lder}}

\def\KBW{\text{KBW}}

\def\fbar{\bar{f}}

\numberwithin{equation}{section}
\theoremstyle{definition}
\newtheorem{thm}{Theorem}[section]
\newtheorem{lem}[thm]{Lemma}

\newtheorem{definition}[thm]{Definition}

\newtheorem{proposition}[thm]{Proposition}
\newtheorem*{theorem*}{Proposition}

\theoremstyle{remark}
\newtheorem{rem}[thm]{Remark}

\allowdisplaybreaks[1]
\newtheorem{innercustomthm}{Condition}
\newenvironment{customthm}[1]
  {\renewcommand\theinnercustomthm{#1}\innercustomthm}
  {\endinnercustomthm}
\newtheorem{innercustomgeneric}{\customgenericname}
\providecommand{\customgenericname}{}
\newcommand{\newcustomtheorem}[2]{  \newenvironment{#1}[1]
  {   \renewcommand\customgenericname{#2}   \renewcommand\theinnercustomgeneric{##1}   \innercustomgeneric
  }
  {\endinnercustomgeneric}
}
\newcustomtheorem{customthm1}{Faithfulness}
\newcustomtheorem{customthm3}{Faithfulness($\delta'$)}
\newcustomtheorem{customthm2}{CS}

\def\KL{\mathsf{KL}}
\def\calD{\mathcal{D}}
\def\calM{\mathcal{M}}

\begin{document}

\begin{frontmatter}

\title{Rejoinder: On nearly assumption-free tests of nominal confidence interval coverage for causal parameters estimated by machine learning}

\author{\fnms{Lin} \snm{Liu}\corref{}\ead[label=e1]{linliu@sjtu.edu.cn}\thanksref{t1}}
\thankstext{t1}{Institute of Natural Sciences and School of Mathematical Sciences, Shanghai Jiao Tong University} 
\author{\fnms{Rajarshi} \snm{Mukherjee}\corref{}\ead[label=e2]{ram521@mail.harvard.edu}\thanksref{t2}}
\thankstext{t2}{Department of Biostatistics, Harvard T. H. Chan School of Public Health} 
\author{\fnms{James M.} \snm{Robins}\corref{}\ead[label=e3]{robins@hsph.harvard.edu}\thanksref{t3}}
\thankstext{t3}{Department of Epidemiology and Biostatistics, Harvard T. H. Chan School of Public Health} 


\end{frontmatter}

We thank the editors for this opportunity and the discussants \citet{kennedy2020discussion} (abbreviated as KBW in the sequel) for their insightful commentaries on our paper \citep{liu2020nearly} (abbreviated as LMR in the sequel).

\section{A brief introduction to higher order influence functions}
\label{sec:review}
We would like to start our rejoinder by responding to the philosophical comments in Section 6 of KBW's discussion before getting into the other more technical comments. In Section 6, KBW divide statistical procedures into structure-driven and methods-driven but also acknowledge that the boundary between these two categories is blurry. For example, even for the poster child of the methods-driven tools -- deep neural networks -- one common research direction is to prove some form of optimality or robustness under some assumptions, often quantified by smoothness, sparsity or other related complexity measures such as metric entropy \citep{schmidt2017nonparametric, hayakawa2020minimax, barron2018approximation}. 

The discussants then state that higher order influence function (HOIF) based methods are `structure-driven' because `they typically rely on carefully constructed series estimates' and achieve `better performance over appropriate \Holder{} spaces potentially at the expense of being more structure driven.' This statement misunderstands the motivation and goals of HOIF estimation. Our goal has always been to make HOIF fully methods-driven. However, before we reach this goal, difficult open problems remain to be solved. Until then, we have had to make restrictive assumptions to obtain sharp mathematical results -- these assumptions can make our methodology appear at least partly `structure-driven'.

The theory of HOIF is (simplifying somewhat) a theory based only on higher order scores of finite dimensional submodels. As a consequence, the theory by itself cannot quantify the rates of convergence of a HOIF estimator and thus the bias of a HOIF estimator without additional complexity reducing model assumptions, a central point we stressed throughout LMR. To be more concrete, for now let us restrict the attention to smooth nonlinear functionals $\psi (\theta)$ of a distribution $\P_{\theta}$ lying in an infinite dimensional model $\calM = \left\{ \P_{\theta}; \theta \in \Theta \right\}$ with a first order influence function $\IIFF_{1, \psi} (\theta)$ but (as is generally the case in infinite dimensional models) without $m$-th order influence functions for $m > 1$. Therefore, HOIF theory often considers finite $k = k(n)$-dimensional sieves $\calM_{sub, k} = \left\{ \P_{\theta}; \theta \in \Theta_{sub, k} \subset \Theta \right\}$ containing an initial training sample estimator $\hat{\theta}$, an associated projection map $\theta \mapsto \tilde{\theta}_{k}$ from $\Theta$ onto $\Theta_{sub, k}$ that is the identity for $\theta \in \Theta_{sub, k}$. The projection map defines a truncated parameter $\tilde{\psi}_{k} (\theta), \theta \in \Theta$ by $\tilde{\psi}_{k} (\theta) = \psi (\tilde{\theta}_{k} (\theta)), \theta \in \Theta$, which will typically have HOIFs of all orders because $\Theta_{sub, k}$ is finite dimensional. The theory of HOIF applied to the parameter $\tilde{\psi}_{k} (\theta)$ guarantees that $\{\tilde{\psi}_{k} (\hat{\theta}) + \BE_{\theta} [\IIFF_{m, \tilde{\psi}_{k}} (\hat{\theta})]\} - \tilde{\psi}_{k} (\theta) = O (\Vert \hat{\theta} - \theta \Vert^{m + 1})$ or, equivalently, 
\begin{eqnarray*}
\BE_{\theta} \left[ \hat{\psi}_{m, k} - \tilde{\psi}_{k} (\theta) \right] &\equiv& \EB_{\theta, m, k} (\hat{\psi}_{1}) = O \left( \left\Vert \hat{\theta} - \theta \right\Vert^{m + 1} \right) \\
\text{ where } \hat{\psi}_{m, k} &=& \psi (\hat{\theta}) + \IIFF_{m, \tilde{\psi}_{k}} (\hat{\theta}).
\end{eqnarray*}
Here $\hat{\psi}_{1} = \psi (\hat{\theta}) + \IIFF_{1, \tilde{\psi}_{k}} (\hat{\theta})$ is a doubly robust machine learning (DRML) estimator based on the first order influence function and $\IIFF_{m, \tilde{\psi}_{k}} (\hat{\theta}) = \IIFF_{1, \tilde{\psi}_{k}} (\hat{\theta}) - \sum_{j = 2}^{m} \IIFF_{j j, \tilde{\psi}_{k}} (\hat{\theta})$ where, under $\P_{\hat{\theta}}$, $\IIFF_{j j, \tilde{\psi}_{k}} (\hat{\theta}) \equiv \widehat{\IIFF}_{j j, k}$ is a $j$-th order $U$-statistic\footnote{Here we are using the same sign convention as in LMR, which reverses the sign conventions of \citet{robins2008higher}.}. Unless stated otherwise all expectations are conditional on the training sample. Thus $\BE_{\theta} [\hat{\psi}_{m, k} - \psi (\theta)] = \EB_{\theta, m, k} (\hat{\psi}_{1}) + \TB_{\theta, k} (\hat{\psi}_{1})$ with $\TB_{\theta, k} (\hat{\psi}_{1}) \equiv \tilde{\psi}_{k} (\theta) - \psi (\theta)$. Furthermore, it is often the case that $\var_{\theta} (\hat{\psi}_{m, k}) = O (k^{m - 1} / n^{m} \vee 1 / n)$. The above is pretty much the cornerstone of the theory of HOIF estimators. This theory involves no structural assumptions on components of $\theta$, such as smooth or sparse nuisance functions. As a consequence the theory is agnostic as to the rate at which $\Vert \hat{\theta} - \theta \Vert$ or $\TB_{\theta, k(n)} (\hat{\psi}_{1})$ converges to zero.

The above theory was introduced in Sections 2-3 of \cite{robins2008higher} before either \Holder{} smoothness assumptions or best approximating bases were introduced. However, we then went on to study models defined in terms of \Holder{} smoothness  to determine whether our abstract theory (just described) could be used to construct rate minimax estimators (it could) for a particular class of functionals under this well known infinite dimensional model. Under the \Holder{} model, we could determine the rate at which $\EB_{\theta, m, k} (\hat{\psi}_{1})$ and $\TB_{\theta, k} (\hat{\psi}_{1})$ converged to $0$ for different choices of the sequences $m = m(n), k = k(n)$, and parametric submodels $\calM_{sub, k}$. We could thus optimize $m(n), k(n)$, and $\calM_{sub, k(n)}$ and often obtained minimax rates under the \Holder{} model, when we did so.

Indeed, our theoretical work on HOIF since \cite{robins2008higher} and \cite{robins2017minimax} can be understood as having been solely directed toward the elimination of remaining structure-driven assumptions. As an example, consider the parameter $\psi (\theta) = \BE_{\theta} [\var_{\theta} [A | X]]$ with $A$ Bernoulli and $X$ high-dimensional with a continuous distribution. Then $\theta = (p, g)$ where $p(x) = \BE_{\theta} [A | X = x]$ and $g(x)$ is the density of $X$. A $k$-dimensional submodel $p(x; \theta_{k}) \equiv p_{\theta_{k}} (x)$ for $p(x)$ was chosen to be $\left\{ p_{\theta_{k}}; p_{\theta_{k}} (x)= \hat{p} (x) + \theta_{k}^{\top} \zbar_{k} (x)\right\}$ where $\zbar_{k} (x)$ is the vector of the first $k$ elements of a sequence of user-selected set of basis functions $\left\{ \zbar_{j} (x), j = 1, \ldots \right\}$. The corresponding projection map is $\tilde{\theta}_{k} = - \Omega_{k}^{-1} \BE_{\theta} [A (\hat{p}(X) - p(X))]$ with $\Omega_{k} \coloneqq \BE_{\theta} [\zbar_{k}(X) \zbar_{k}(X)^{\top}] \equiv \BE_{g} [\zbar_{k}(X) \zbar_{k}(X)^{\top}]$. The HOIFs $\IIFF_{m, \tilde{\psi}_{k}} (\theta)$ depended on $g$ through $\Omega_{k}$ which we estimated by $\BE_{\hat{g}} [\zbar_{k}(X) \zbar_{k}(X)^{\top}]$ with $\hat{g}$ an estimator of the density $g$. In the above papers, we used complexity reducing models (e.g. \Holder{} models) on both $p$ and $g$ to evaluate the rate of convergence of $\EB_{\theta, m, k} (\hat{\psi}_{1}) = O (\Vert \hat{p} - p \Vert^{2} \Vert \hat{g} - g \Vert^{m - 1})$ to zero. In the case in which $k = o(n)$, \cite{mukherjee2017semiparametric} introduced ``empirical'' HOIF estimators that eliminated the need to assume a complexity-reducing model on $g$. Instead they proposed estimating $\{\BE_{g} [\zbar_{k}(X) \zbar_{k}(X)^{\top}]\}^{-1}$ by the inverse sample Gram matrix $\{\widehat{\Omega}_{k}^{tr}\}^{-1} \equiv \{\BP_{n_{tr}} [\zbar_{k}(X) \zbar_{k}(X)^{\top}]\}^{-1}$ in the training sample for $k < n$. Indeed, the goal of LMR was to determine the inferential questions concerning smooth nonlinear functionals that remain answerable when one refuses to impose any complexity reducing structural assumptions -- a goal that seems to us to be extremely ``methods-driven''.

However, several difficult open problems remain to be solved before HOIF inference becomes fully methods-driven; i.e. becomes a robust, off-the-shelf, widely applicable methodology for inference on non-linear functionals in non- and semi-parametric models. We have previously discussed these remaining problems both in LMR and earlier papers. In this rejoinder we discuss some of them in greater depth to respond to discussants' concerns and suggestions.

\section{Towards ``methods-driven'' HOIF\lowercase{s}} 
\label{sec:issues}
The main bottleneck in achieving fully ``methods-driven'' HOIF inference is the dependence of the power of our falsification tests and the efficiency of our estimators on the choice of the basis functions $\zbar_{k} (x)$. In Section \ref{sec:basis_choice}, we propose a relaxation of one of the assumptions in LMR that dispenses with our reliance on `carefully constructed'  choices (such as compactly supported wavelets or B-splines) for the basis functions $\zbar_{k} (x)$, at the cost of perhaps a small loss in power. 

In their final section, KBW consider one of the most interesting open problem in the theory of HOIF: how to adaptively select the $m$ basis functions $\fbar_{m} (x) = \left( f_{1} (x), \ldots, f_{m} (x) \right)$ of the $d$-dimensional vector $x$ to (approximately) minimize the truncation bias for the expected conditional variance\footnote{Following KBW, we have used $m$ rather than $k$ to indicate the dimension of the vector of basis functions in their statistic $\widehat{\IIFF}_{22, \KBW}$ defined in Section \ref{sec:basis_selection}. Note that by Theorem 3.2 of LMR, we require $m < n$ to have power to reject the null hypothesis under the alternative that $\frac{\Bias_{\theta, k} (\hat{\psi}_{1})}{\se_{\theta} (\hat{\psi}_{1})} = \delta + c$ for any given $c > 0$. Recall that our ability to detect with probability going to 1 any alternative of order $c / n^{1/2}$ for any fixed $c > 0$ is one of the perhaps surprising consequences of our tests based on HOIF (due to the fact that with $k < n$, the variance of $\widehat{\IIFF}_{22, k}$ under that alternative is $k / n^{2}$.)}. We have been investigating this same problem for several years but we have yet to come up with a wholly satisfying approach. KBW suggest a new approach based on aggregation. In Section \ref{sec:basis_selection}, we show by a toy example that this approach seems promising and is worth further in-depth investigation. However, we also raise a difficult problem that needs to be solved before this promise can be fulfilled. 

\subsection{Dispensing with the need for carefully constructed basis functions}
\label{sec:basis_choice}
LMRs assumed Condition W in the statement and proof of Theorem 3.2 and 4.2. Condition W imposes severe restrictions on the basis function $\zbar_{k}$ that can be chosen. Here we show these restrictions can be avoided by replacing Condition W with Condition \ref{cond:w} below. In fact, we mentioned Condition \ref{cond:w} in Remark 2.5 of LMR but failed to provide sufficient emphasis and context. With the exception of the online supplement, following the recommendation of a referee, LMR focus on the semisupervised setting in which $\BE_{g} [\zbar_{k}(X) \zbar_{k}(X)^{\top}]$ is known. In that case we only require the following weakened form of Condition W in LMR for the level and power properties stated in Theorem 3.2 and 4.2 of our test $\widehat{\chi}_{k} (z_{\alpha / 2}, \delta)$ for the surrogate null hypothesis $\H_{0, k} (\delta): \frac{\vert \Bias_{\theta, k} (\hat{\psi}_{1}) \vert}{\se_{\theta} [\hat{\psi}_{1}]} \leq \delta$ to hold, where $\Bias_{\theta, k} (\hat{\psi}_{1}) = \BE_{\theta} [\hat{\psi}_{1} - \tilde{\psi}_{k} (\theta)]$. 

\begin{customthm}{SW}\leavevmode\label{cond:w}
\begin{enumerate}
\item All the eigenvalues of $\Omega_{k}$ are bounded away from 0 and $\infty$;

\item The true nuisance functions $b(X)$ and $p(X)$, and the estimated nuisance functions $\hat{b}(X)$ and $\hat{p}(X)$, are all bounded with $\P_{\theta}$-probability 1;

\item $\Vert \zbar_{k}(x)^{\top} \zbar_{k}(x) \Vert_{\infty} \leq B k$ for some constant $B > 0$.
\end{enumerate}
\end{customthm}

Condition \ref{cond:w} weakens Condition W in LMR by dropping the requirement that $\Vert \Pi [\hat{b} - b | \zbar_{k}] \Vert_{\infty} \leq C$ and $\Vert \Pi [\hat{p} - p | \zbar_{k}] \Vert_{\infty} \leq C$ for some constant $C > 0$ not depending on $n$. This extra condition holds for wavelets, B-spline and local polynomial partition series \citep{belloni2015some}. However, there are many additional choices of  $\zbar_{k}$ that satisfy Condition \ref{cond:w} without satisfying Condition W, including Fourier series and monomial transformations of the covariates $X$ when $X$ is compactly supported or monomial transformations of some bounded transformation of $X$ when $X$ is unbounded. Allowing $\zbar_{k}$ to include monomial transformations of the covariates makes our approach more flexible and ``methods-driven''.

Turn now to the case considered in the supplement of LMR and \cite{liu2020skepticism} in which the expected Gram matrix $\Omega_{k} = \BE_{g} [\zbar_{k}(X) \zbar_{k}(X)^{\top}]$ is unknown and therefore tests of $\H_{0, k} (\delta)$ must now be based on empirical HOIFs that substitute $\widehat{\Omega}_{k} = \mathbb{P}_{n_{tr}} [\zbar_{k}(X) \zbar_{k}(X)^{\top}]$ for $\Omega_{k}$. In Section \ref{app:emp} (also in Section S3 of the supplement of LMR and \citet[Sections 4 and 5]{liu2020skepticism}), we show in that case that the additional conditions $\Vert \Pi [\hat{b} - b | \zbar_{k}] \Vert_{\infty} \leq C$ and $\Vert \Pi [\hat{p} - p | \zbar_{k}] \Vert_{\infty} \leq C$ are needed for tests $\widehat{\chi}_{33, k} (\widehat{\Omega}_{k}^{-1}; z_{\alpha / 2}, \delta)$ (see equation \eqref{eq:test}) that use $\widehat{\Omega}_{k}$ to attain the same asymptotic power as the oracle tests $\widehat{\chi}_{k} (z_{\alpha / 2}, \delta)$ that use $\Omega_{k}$. When these infinity-norm bounds do not hold (e.g. for Fourier series or monomial transformation of compactly-supported covariates), the asymptotic power of the test might be less. However, the level of the test $\widehat{\chi}_{33, k} (\widehat{\Omega}_{k}^{-1}; z_{\alpha / 2}, \delta)$ under Conditions W and \ref{cond:w} are identical under some additional restrictions\footnote{For the additional restrictions, see Proposition \ref{prop:level} and Remark S3.5 of the supplement of LMR. In \citet{liu2020skepticism}, we show that it is possible to remove these \textit{additional restrictions} by extending $U$-statistics of order three to diverging order with increased computational cost.}. Hence, at least in the context of bias testing, we shall relax Condition W to Condition \ref{cond:w} in the future (e.g. \cite{liu2020skepticism}), whether or not $\BE_{g} [\zbar_{k}(X) \zbar_{k}(X)^{\top}]$ is known, so as to remove restriction to ``carefully constructed series''. The only cost is a possible small loss in power and that only if the infinity norms of the projections $\Pi [\hat{b} - b | \zbar_{k}]$ and/or $\Pi [\hat{p} - p | \zbar_{k}]$ are not bounded even though those of $\hat{b} - b$ and $\hat{p} - p$ are bounded under Condition \ref{cond:w}. In summary, by adopting Condition \ref{cond:w} rather than Condition W, our methodology becomes more ``methods-driven'' and allows $\zbar_{k}$ to include most basis functions relevant for practice.

\subsection{KBW's aggregation approach}
\label{sec:basis_selection}
In the following, to avoid irrelevant issues, we consider the case in which the density $g$ of $X$ is known, $\psi (\theta) = \BE_{\theta} [\var_{\theta} [A | X]]$, and our inferential goal is to test the actual null hypothesis $\H_{0} (\delta): \Bias_{\theta} (\hat{\psi}_{1}) \leq \delta \se_{\theta} (\hat{\psi}_{1})$ as in LMR.

KBW propose the following procedure. First divide the data into three randomly selected subsamples: a training sample $\calD_{tr}$, a selection (auxiliary) sample $\calD_{sel}$, and an estimation sample $\calD_{est}$. We are given an estimate $\hat{p} (x)$ of $\BE_{\theta} [A | X = x]$ obtained from $\calD_{tr}$. We then use data $\calD_{sel}$ to regress the residuals $A - \hat{p} (X)$ using $m$ different methods to obtain $\hat{\fbar}(x) = \{\hat{f}_{\ell} (x), \ell = 1, \ldots, m\}$ predictors of the true residual function $p (x) - \hat{p} (x)$. Finally in sample $\calD_{est}$, we compute\footnote{As in LMR, we always condition on $\mathcal{D}_{tr}$, which is therefore suppressed in the notation.} 
\begin{equation} \label{kbw}
\widehat{\IIFF}_{22, \KBW} (\hat{\fbar}_{m}) = \frac{1}{n (n - 1)} \sum_{i_{1} \neq i_{2} \in \calD_{est}} (A_{i_{1}} - \hat{p}(X_{i_{1}})) \hat{\fbar}_{m} (X_{i_{1}})^{\top} \Omega_{\fbar_{m}}^{-1} \hat{\fbar}_{m} (X_{i_{2}}) (A_{i_{2}} - \hat{p}(X_{i_{2}})).
\end{equation}
where $\Omega_{\fbar_{m}} \equiv \BE_{g} [\hat{\fbar}_{m}(X) \hat{\fbar}_{m}(X)^{\top} | \calD_{sel}] = \int \hat{\fbar}_{m}(x) \hat{\fbar}_{m}(x)^{\top} g (x) d x$. KBW test the hypothesis $\H_{0} (\delta)$, using their test statistic
$$
\widehat{\chi}_{\KBW, m} (z_{\alpha / 2}, \delta) = \mathbbm{1} \left\{ \frac{\widehat{\IIFF}_{22, \KBW} (\hat{\fbar}_{m})}{\widehat{\se} (\hat{\psi}_{1})} - z_{\alpha / 2} \frac{\widehat{\se} (\widehat{\IIFF}_{22, \KBW} (\hat{\fbar}_{m}))}{\widehat{\se} (\hat{\psi}_{1})} > \delta \right\}
$$
which is precisely the test statistic $\widehat{\chi}_{k} (z_{\alpha / 2}, \delta)$ in Section 3 of LMR except that $\widehat{\IIFF}_{22, k}$ is replaced by $\widehat{\IIFF}_{22, \KBW} (\hat{\fbar}_{m})$.

Roughly speaking, KBW's idea is that if one among the $m$ ($m < n$) methods both captures the true structure class (e.g. a smoothness class versus a sparsity class) of the true residual function $p (x) - \hat{p} (x)$ and appropriately chosen tuning parameters, then the power of the KBW test $\widehat{\chi}_{\KBW, m} (z_{\alpha / 2}, \delta)$ will be equal to or greater than the power of the LMR test $\widehat{\chi}_{k} (z_{\alpha / 2}, \delta)$ that would be used by an oracle who knew the true structure class and the optimal basis functions $\zbar_{k} (x)$ that maximize the power of $\widehat{\chi}_{k} (z_{\alpha / 2}, \delta)$.

Below we will show how the KBW test performs in a toy example under \Holder{} assumptions. 

\underline{{\bf Example}}: Let $X$ be one-dimensional, $p (x)$ be \Holder{} with exponent $s$ with $s < 1 / 2$ and $\hat{p} (x)$ be a series estimator of $p$ using $\zbar_{k^{\ast}}$, the first $k^{\ast} \asymp n^{1 / (1 + 2s)}$ compactly supported Cohen-Daubechies-Vial (CDV) wavelets (of sufficient order). 

This $k^{\ast}$ is minimax optimal for estimating $p$ under mean squared error loss. With this minimax optimal $\hat{p}$, $\Bias_{\theta} (\hat{\psi}_{1}) = \BE_{\theta} [(p(X) - \hat{p}(X))^{2}] \asymp n^{-2s / (1 + 2s)}$. Since $s < 1 / 2$, $\Bias_{\theta} (\hat{\psi}_{1}) > n^{-1/2}$ and thus $\H_{0} (\delta)$ is false. For any $k$ such that $n > k \gg k^{\ast}$, $\Bias_{\theta, k} (\hat{\psi}_{1}) \equiv \BE_{\theta} [\widehat{\IIFF}_{22, k}] \asymp \Bias_{\theta} (\hat{\psi}_{1})$ (or equivalently $\frac{\TB_{\theta, k} (\hat{\psi}_{1})}{\Bias_{\theta} (\hat{\psi}_{1})} \rightarrow 0$) and by Theorem 3.2 of LMR, the LMR test $\widehat{\chi}_{k} (z_{\alpha / 2}, \delta)$ rejects $\H_{0} (\delta)$ with probability approaching 1 with increasing sample size.


In the same setup, we consider the KBW test $\widehat{\chi}_{\KBW, m} (z_{\alpha / 2}, \delta)$ with $m = 1$ and $\hat{f}$ the estimated regression function of $A - \hat{p}(X)$ on $\zbar_{k} (X)$, with $k^{\ast} \ll k \ll \left( \frac{k^{\ast}{}^{2}}{\sqrt{n}} \wedge n \right)$. Below we argue that the KBW test $\widehat{\chi}_{\KBW, m = 1} (z_{\alpha / 2}, \delta)$ with such $\zbar_{k}$ rejects $\H_{0} (\delta)$ with probability approaching 1. The reason for $k \ll \frac{k^{\ast}{}^{2}}{\sqrt{n}}$ will be explained below. Thus in summary, for every $s \in (0, 1 / 2)$, we can construct a test $\widehat{\chi}_{k} (z_{\alpha / 2}, \delta)$ as in LMR and a test $\widehat{\chi}_{\KBW, m = 1} (z_{\alpha / 2}, \delta)$ as in KBW to reject $\H_{0} (\delta)$ with probability approaching 1.

We now complicate our example by supposing that we only know $s \in (0, 1/2)$ but do not know its true value. Therefore now $\hat{p}$ will be an optimal adaptive estimator of $p$ estimated from the training sample $\calD_{tr}$. Our goal now is to construct an adaptive test that rejects with probability going to 1 whatever the true value $s$ is. The LMR test $\widehat{\chi}_{k} (z_{\alpha / 2}, \delta)$ with $k = n / c$ with any $c > 1$ is adaptive in this sense. However, no KBW test $\widehat{\chi}_{\KBW, m = 1} (z_{\alpha / 2}, \delta)$ can be adaptive. To see this, if $s = 1 / 2 - \epsilon$ for some very small $\epsilon > 0$, then the value of $k$ used for the test $\widehat{\chi}_{\KBW, m = 1} (z_{\alpha / 2}, \delta)$ must be less than $\frac{k^{\ast}{}^{2}}{\sqrt{n}} = n^{\frac{1}{2} + \epsilon}$ with $k^{\ast}$ evaluated at $n^{\frac{1}{1 + 2 (1 / 2 - \epsilon)}}$. On the other hand, if $s = \epsilon$, $k$ must be larger than $k^{\ast} = n^{\frac{1}{1 + 2\epsilon}}$. Hence there exists no single test $\widehat{\chi}_{\KBW, m = 1} (z_{\alpha / 2}, \delta)$ that can adapt. It remains an open question whether there exists an adaptive test $\widehat{\chi}_{\KBW, m} (z_{\alpha / 2}, \delta)$ for any choice of $m$, even in this simple example.

We now justify the above claims. KBW first estimate the regression of $A - \hat{p}(X)$ on $\zbar_{k} (X)$ from the selection sample $\mathcal{D}_{sel}$ as follows: 
\begin{equation*}
\hat{f} (x) = \zbar_{k} (x)^{\top} \hat{\beta}
\end{equation*}
where $\hat{\beta} = \Omega_{k}^{-1} \frac{1}{n} \sum_{i \in \calD_{sel}} \zbar_{k} (X_{i}) (A_{i} - \hat{p} (X_{i}))$. Then KBW compute $\widehat{\IIFF}_{22, \KBW} (\hat{f})$ as in equation \eqref{kbw}.

Recall that by choosing $k \gg k^{\ast}$, $\Bias_{\theta, k} (\hat{\psi}_{1}) \asymp \Bias_{\theta} (\hat{\psi}_{1})$. Following the proof of Theorem 3.2 of LMR, for $\widehat{\chi}_{\KBW, 1} (z_{\alpha / 2}, \delta)$ to reject the null hypothesis $\H_{0} (\delta)$ with probability approaching 1, we need
\begin{enumerate}[label=(\roman*)]
\item the mean of $\widehat{\IIFF}_{22, \KBW} (\hat{f})$ exceeds $n^{-1/2}$: $\BE_{\theta} [\widehat{\IIFF}_{22, \KBW} (\hat{f})] \asymp \frac{\Bias_{\theta, k} (\hat{\psi}_{1})^{2}}{k / n} \gg \se_{\theta} (\hat{\psi}_{1}) \asymp \frac{1}{\sqrt{n}}$; and
\item the mean of $\widehat{\IIFF}_{22, \KBW} (\hat{f})$ dominates its standard error: $\BE_{\theta} [\widehat{\IIFF}_{22, \KBW} (\hat{f})] \gg \se_{\theta} [\widehat{\IIFF}_{22, \KBW} (\hat{f})]$.
\end{enumerate}

In Section \ref{app:order}, we show that $\widehat{\IIFF}_{22, \KBW} (\hat{f})$ has mean of order $\frac{\Bias_{\theta, k} (\hat{\psi}_{1})^{2}}{k / n}$, which dominates its standard error of order $\frac{\Bias_{\theta, k} (\hat{\psi}_{1}) / \sqrt{n}}{k / n}$ when $\H_{0} (\delta)$ is false and $\Bias_{\theta, k} (\hat{\psi}_{1}) \asymp \Bias_{\theta} (\hat{\psi}_{1})$. Hence (ii) should hold. In terms of (i), if $\H_{0} (\delta)$ is false, i.e. $\Bias_{\theta} (\hat{\psi}_{1}) \asymp \frac{k^{\ast}}{n} \gg \frac{1}{\sqrt{n}}$ or equivalently $k^{\ast} \gg \sqrt{n}$ (or equivalently the smoothness index $s < 1 / 2$), there always exists $k \gg k^{\ast}$ also satisfying (i) as (1) $k^{\ast} \ll \frac{k^{\ast}{}^{2}}{\sqrt{n}}$ when $k^{\ast} \gg \sqrt{n}$ and (2) (i) is equivalent to
\begin{align*}
\frac{\Bias_{\theta, k} (\hat{\psi}_{1})^{2}}{k / n} \asymp \frac{\Bias_{\theta} (\hat{\psi}_{1})^{2}}{k / n} \asymp \frac{k^{\ast}{}^{2}}{k n} \gg \frac{1}{\sqrt{n}} \Leftrightarrow k \ll \frac{k^{\ast}{}^{2}}{\sqrt{n}}.
\end{align*}



\section{Classes of functionals and the bias test}
\label{sec:general}
\subsection{Monotone bias class and alternative sample splitting schemes}
\label{sec:mb} 
The analysis of the expected density example described in Section 2 of KBW illustrates our approach in a simple setting and is essentially isomorphic to the analysis of the expected conditional variance $\BE_{\theta} [\var_{\theta} [A | X]]$ when, as in the semi-supervised setting considered in the main text, the marginal distribution of $X$ is known. Both of these functionals are members of our monotone bias class, a class so named because of our claim that the bias of the our second order $U$-statistic estimator was non-increasing in the number of basis functions $k$. KBW show that our claim does not hold if, following \citet{newey2018cross}, one uses two different estimates of the density of $X$ (or of $\BE_{\theta} [A | X]$ in the $\BE_{\theta} [\var_{\theta} [A | X]]$ example) coming from two independent subsamples of the training sample. This is a fact we were well aware of (see Section S.1.1 of the supplement of LMR). 

The multiple training sample splitting nuisance function estimators of \cite{newey2018cross} can sometimes achieve faster convergence rates than the (single training sample) DRML estimators considered in the main text of LMR. In fact, they can even achieve $n^{-1/2}$ rates for estimation of $\BE_{\theta} [\var_{\theta} [A | X]]$ and $\BE_{\theta} [\cov_{\theta} [A, Y | X]]$ under minimal \Holder{} smoothness assumptions \citep{robins2009semiparametric} needed for $\sqrt{n}$-estimation. But this result requires, not only that the \Holder{} class assumptions are true, but also that one uses particular \textit{undersmoothed} nuisance function estimators (such as local polynomials or regression splines) rather than flexible black-box machine learning (such as random forests or deep learning) estimators, that are the motivation for and the subject of our paper. 

Indeed, it was our intention to define DRML estimators to be estimators in which all nuisance functions were estimated with ML algorithms from the same training sample, as this is the current `state of the art'; and, as we emphasized in the Introduction of our paper, our approach is one of being ``in dialogue with current practice and practitioners." In fact, our technical Lemma 2.3 of LMR, which is the Lemma in which we differentiate the bias properties of $\BE_{\theta} [\var_{\theta} [A | X]]$ from those of $\BE_{\theta} [\cov_{\theta} [A, Y | X]]$, is correct as stated, as the Lemma explicitly applies to the case in which a single estimator of $\hat{p}$ of $\BE_{\theta} [A | X]$ is used in the DRML estimator $\hat{\psi}_{1}$ of $\BE_{\theta} [\var_{\theta} [A | X]]$ (as is evident from the paragraph just prior to the Lemma). 

There remain important open problems that we, with other colleagues, and our discussants are considering in ongoing work on the properties of multiple training sample splitting nuisance function estimators: (1) is it possible to develop a general theory of the benefit of multiple training sample splitting, even when black-box machine learning estimators of the nuisance functions are employed and (2) can undersmoothing be automated to directly balance bias with variance for the estimators of the functional of interest? The model selection strategy in \cite{cui2019bias} may be a viable option.

\subsection{Generalization to functionals with the mixed bias property}
\label{sec:dr}
In \citet{liu2020skepticism}, we also show that the methods of LMR and the theory of HOIFs \citep{robins2008higher, robins2016technical} can be extended to the entire class of parameters/functionals with the so-called \textit{mixed bias property} (henceforth called MB functionals), studied by \cite{rotnitzky2019characterization}. This class is a strict superset of the union of two overlapping classes of functionals introduced in \cite{robins2008higher} and \cite{chernozhukov2018learning} respectively:
\begin{definition}[Definition 1 of \protect\cite{rotnitzky2019characterization}]
\label{def:dr}\leavevmode
A parameter/functional $\psi (\theta)$ is an MB functional if, for each $\theta \in \Theta$, there exist $b: x \mapsto b(x) \in \mathcal{B}$ and $p: x \mapsto p(x) \in \mathcal{P}$ such that (i) $\theta = (b, p, \theta_{\backslash (b, p)})$ and $\Theta = \mathcal{B} \times \mathcal{P} \times \Theta_{\backslash (\mathcal{B}, \mathcal{P})}$ and (ii) for any $\theta, \theta'$ 
\begin{equation}
\psi (\theta') - \psi (\theta) + \BE_{\theta} \left[ \IF_{1} (\theta') \right] = \BE_{\theta} \left[ S_{bp} (b(X) - b' (X)) (p(X) - p' (X)) \right]  \label{eq:drbias}
\end{equation}
where $S_{bp} \equiv s_{bp} (O)$ and $o \mapsto s_{bp} (o)$ is a known function that does not depend on $\theta$ or $\theta'$ satisfying either $\P_{\theta} (S_{bp} \geq 0) = 1$ or $\P_{\theta} (S_{bp} \leq 0) = 1$ and $\mathsf{IF}_{1} (\theta)$ is the (first order) influence function of the parameter $\psi (\theta)$.
\end{definition}

\section{On the covariate structure}
\label{sec:structure} 
KBW raise the interesting question of ``how and whether randomness of the covariates and/or smoothness of the covariate density should be relied on in practice''. They suggest that, perhaps, the covariates $X$ should be conditioned on (and thus be regarded as fixed rather than random) in any inferential procedure, whenever study subjects are not randomly sampled from some well-defined population. If this suggestion were followed the benefits of HOIFs may be greatly diminished, but, if so, HOIFs would not be alone. In many areas, it is essential that the covariates $X$ and the outcome $Y$ have a joint distribution $(X, Y)$ -- conformal inference, prediction risk minimization, the bias-corrected lasso, and covariate shift methods being four current examples, at least two of which our discussants have written about. In fact, if the subjects were not randomly sampled, one could equally ask why should we consider $Y$ conditioned on $X = x$ as random, as we have no reason to believe any measure of their association is invariant across studies or populations, especially if non-causal.

So the question is how to inject randomness into an observational study, a question that arises when an investigator wishes to generalize her findings from the observed study subjects to some larger population. For example, an investigator who considers recommending a public health intervention, would hope to have studied subjects that are in some sense representative of the population of potential recipients. The simplest random model allowing generalization is to consider the study subjects as a random sample from some very large (effectively infinite) hypothetical (i.e. fictitious) superpopulation of potential recipients with the superpopulation empirical distribution serving as the target of inference. This is effectively reverse engineering, in the sense that even if one has studied a `convenience sample', one can still hypothesize a superpopulation that is similar, up to sampling variability, to your study subjects. This construction allows the use of ordinary i.i.d. statistical methods, to obtain valid confidence intervals for functionals of the superpopulation empirical distribution [see \cite{robins1988confidence}]. This approach may seem distasteful (or even vacuous) to a purist and not as elegant as de Finetti's subjectivist approach, but we believe it underlies what frequentist analyses that epidemiologists and statisticians are doing daily, perhaps without explicit recognition. Of course, generalization from a convenience sample to an actual, non-hypothetical, population is possible only based on further substantive knowledge. Typical examples of the sort of convenience samples we are thinking of include (i) all members of a HMO admitted for an acute myocardial infarction between Jan 1, 2015 and Dec 31, 2017 or (ii) all workers employed at a particular asbestos mine at any point in the interval 1959 to 1965.

Putting such philosophical matters aside, let us for now assume $X$ is random. In this case, KBW raised the question whether the performance of HOIF tests and estimators is robust to misspecification of a model for the density $g$ of X say in the estimation of $\BE_{\theta} [\var_{\theta} [A | X]]$. The answer is HOIFs are not robust as can be seen from the fact that the estimation bias $\mathbb{E}_{\theta} [\hat{\psi}_{m, k} - \tilde{\psi}_{k} (\theta)] \equiv \mathsf{EB}_{\theta, m, k} (\hat{\psi}_{1}) = O (\Vert \hat{p} - p \Vert^{2} \Vert \hat{g} - g \Vert^{m - 1})$. Thus if $\hat{g}$ is inconsistent with $\Vert \hat{g} - g \Vert > 1$, then $\hat{\psi}_{m, k}$ may have bias even greater than the bias $O (\Vert \hat{p} - p \Vert^{2})$ of the DRML estimator $\hat{\psi}_{1}$. If, following KBW's suggestion, we react to this non-robustness by choosing to make no assumptions whatsoever regarding the density of X, then we need to restrict ourselves to the use of empirical higher order influence functions, as they do not require that we estimate $g$ \citep{mukherjee2017semiparametric}. The following example helps understand their statistical properties.  

Following KBW, suppose the functional of interest is $\psi (\theta) = \mathbb{E}_{\theta} [\mathsf{b}_{1}(X)] = \mathbb{E}_{\theta} [A\mathsf{p}(X)Y]$, where $\psi (\theta )$ is equal to the counterfactual mean of an outcome $Y$ when a binary treatment $A$ is set to $a = 1$ under ignorability of treatment $A$ conditional on the $d$-dimensional covariate $X$. Here $\mathsf{b}_{a}(x) = \mathbb{E}_{\theta} [Y | A = a, X = x]$ and $\mathsf{p} (x) = 1 / \mathbb{E}_{\theta} [A | X = x]$. Suppose we have a correct \text{H\"{o}lder}{} smoothness model for $\mathsf{b}_{1}(x)$ and $\mathsf{p}(x)$ having unknown smoothness exponents $s_{\mathsf{b}}$ and $s_{\mathsf{p}}$ unrestricted except for the requirement that $s / d \equiv (s_{\mathsf{b}} + s_{\mathsf{p}}) / (2d) < 1 / 4$, implying $\psi (\theta)$ cannot be estimated at rate $n^{-1/2}$ \citep{robins2009semiparametric}. Suppose however we make no complexity reducing assumptions about $f_{X}(x)$, the density of $X$. In that case we would need to use an empirical HOIF estimator $\hat{\psi}_{m, k}^{emp} = \psi (\hat{\theta}) + \mathbb{IF}_{m,\tilde{\psi}_{k}}^{emp} (\hat{\theta})$ with $\bar{\mathsf{z}}_{k}(x)$ chosen as appropriate $d$-dimensional compactly supported CDV wavelets for optimal approximation with $k < n_{tr} = n / 2$\footnote{$n_{tr}$ denotes the training sample size.}, so the inverse of empirical Gram matrix $\mathbb{P}_{n_{tr}} [\bar{\mathsf{z}}_{k}(X) \bar{\mathsf{z}}_{k}(X)^{\top}]$ exists with high probability. It follows from a slight modification of Theorem 5 of \citet{mukherjee2017semiparametric} that with $k = n / (\mathsf{log}(n))^{3}$ and $m = \sqrt{\mathsf{log}(n)}$, $\hat{\psi}_{m, k}^{emp}$ has, up to log terms, truncation bias $n^{- 2 s / d}$ and variance $O(1 / n)$ and negligible estimation bias and thus attains (up to log terms), the rate of convergence $n^{- 2 s / d}$ found by \citet{wang2008effect} for $\mathbb{E}_{\theta} [\mathsf{var}_{\theta} [Y | X]]$ under a fixed design. We believe that no other estimator of the counterfactual mean is known to achieve this rate of convergence for all $s_{\mathsf{b}}$ and $s_{\mathsf{p}}$ satisfying $(s_{\mathsf{b}} + s_{\mathsf{p}}) / (2d) < 1/4$ without imposing further assumptions on the density of $X$. Finally, we conjecture the log terms may be eliminated by decreasing the rate at which $k / n$ converges to zero and appropriately adjusting the rate at which $m(n)$ increases with $n$.

Next we turn to KBW's question concerning the performance of HOIF when $X$ is random but we perform inference conditional on $X$. We begin by providing an example in which unconditional and conditional inference are essentially equivalent. \citet{robins2008higher} consider the question of whether, for estimation of a conditional variance, random regressors provide faster rates of convergence than do fixed equal-spaced regressors, and, if so, how? They consider a setting in which $n$ i.i.d.~copies of $(Y, X)$ are observed with $X$ a $d$-dimensional random vector, with density $f (\cdot)$ bounded away from 0 and infinity and absolutely continuous w.r.t.~the uniform measure on the unit cube $[0, 1]^{d}$. The regression function $b(x) = \mathbb{E}_{\theta} [Y | X = x]$ is assumed to lie in a given H\"{o}lder ball with H\"{o}lder exponent $s < 1$. The goal is to estimate $\mathbb{E}_{\theta} [\var_{\theta} [Y | X]]$ under the homoscedastic semiparametric model $\var_{\theta} [Y | X] = \sigma^{2}$. Under this model, the authors construct a simple estimator $\widehat{\sigma}^{2}$ that converges at rate $n^{-\frac{4 s / d}{1 + 4 s / d}}$, when $s / d < 1 / 4$. \citet{shen2019optimal} recently proved this estimator was minimax optimal by proving a matching lower bound\footnote{Suppose we change the model by assuming $s > 1$ but in no other way. In that case, when $s / d < 1 / 4$ no estimator that attains the rate $n^{-\frac{4s / d}{1 + 4s / d}}$ is known; however no lower bound has been proved that would imply such an estimator is impossible.}.

\citet{wang2008effect} and \citet{cai2009variance} earlier proved that if $X_{i}, i = 1, \ldots, n$, are non-random but equally spaced in $[0, 1]^{d}$, the minimax rate of convergence for the estimation of $\sigma^{2}$ is $n^{-2s / d}$ (when $s / d < 1 / 4$) which is slower than $n^{-\frac{4 s / d}{1 + 4 s / d}}$. Thus randomness in $X$ allows for improved convergence rates even though no smoothness assumptions are made regarding $f (\cdot)$. 

To explain how this happens we describe the estimator of \citet{robins2008higher}. The unit cube in $\mathbb{R}^{d}$ is divided into $k = k(n) = n^{\gamma}$, $\gamma > 1$ identical sub-cubes each with edge length $k^{-1 / d}$. For each sub-cube with two or more observations, we randomly select two subjects $i$ and $j$ without replacement. We estimate $\sigma^{2}$ in each such sub-cube by $(Y_{i} - Y_{j})^{2} / 2$. Our estimator $\widehat{\sigma}^{2}$ of $\sigma^{2}$ is the average of the sub-cube-specific estimates $(Y_{i} - Y_{j})^{2} / 2$ over all the sub-cubes with at least two observations. 

A simple probability calculation shows that the number of sub-cubes containing at least two observations is $O_{\P_{\theta}} (n^{2} / k)$ so $\widehat{\sigma}^{2}$ has variance of order $k / n^{2}$ conditional on $\mathbb{X} = (X_{1}, \ldots, X_{n})$. 

To compute the conditional bias, observe that for a given sub-cube with $i$ and $j$ selected $\mathbb{E}_{\theta} [(Y_{i} - Y_{j})^{2} / 2 | \mathbb{X}] = \mathbb{E}_{\theta} [(Y_{i} - Y_{j})^{2} / 2 | X_{i}, X_{j}] = \sigma^{2} + \{b(X_{i}) - b(X_{j})\}^{2} / 2$. However, $\left\vert b(X_{i}) - b(X_{j}) \right\vert = O (\Vert X_{i} - X_{j} \Vert^{s})$ as $s < 1$ and $\Vert X_{i} - X_{j} \Vert = d^{1 / 2} O (k^{-1 / d})$ as $X_{i}$ and $X_{j}$ are in the same sub-cube. It follows that the conditional bias is $\mathbb{E}_{\theta} [\widehat{\sigma}^{2} - \sigma^{2} | \mathbb{X}] = O_{\P_{\theta}} (k^{-2 s / d})$. We next find the $k$ that equates variance and squared bias. Specifically we solve $k / n^{2} = k^{-4 s / d}$ which gives $k = n^{\frac{2}{1 + 4 s / d}}$. Then the rate of convergence at this optimal $k$ is $n^{-\frac{4 s / d}{1 + 4 s / d}}$ conditional on $\mathbb{X}$ with high probability, since $n^{-\frac{4 s / d}{1 + 4 s / d}}$ is $(k / n^{2})^{1/2}$ evaluated at $k = n^{\frac{2}{1 + 4 s / d}}$. But $n^{-\frac{4 s / d}{1 + 4 s / d}}$ is also the unconditional rate of convergence since the conditional bias and conditional variance are of order $O(k^{- 2 s / d})$ and $O (k / n^{2})$ with overwhelming probability\footnote{To see this, define the event $I_{c}$, for any $c > 0$, as
\begin{align*}
I_{c} \coloneqq \left\{ \exists \geq c \frac{n^{2}}{k} \text{ bins with at least two observations of $X$'s} \right\}
\end{align*}
On this event, we have $\left( \hat{\sigma} - \sigma \right)^{2} \lesssim n^{- 4s / (4s + d)}$ with high probability. Furthermore, for $k \ll n^{2}$, there exists $c > 0$ such that $I_{c}$ happens with high probability when the density $f_{X}$ is bounded away from 0 and $\infty$.}.

\citet{robins2008higher} conclude that the random design estimator has better bias control and hence converges faster than the optimal equal-spaced fixed design estimator, because the random design estimator exploits the $O_{\P_{\theta}} (n^{2} / n^{\frac{2}{1 + 4 s / d}})$ random fluctuations for which the $X$'s corresponding to two different observations are only a distance of $O(\{n^{\frac{2}{1 + 4 s / d}}\}^{-1 / d})$ apart. In summary, our calculations conditional on $\mathbb{X}$ indicate that the difference in rates is wholly attributable to the difference between the empirical distribution of the $X_{i}, i = 1, \ldots, n$ in a typical realization under the random design and the empirical distribution in the fixed design. That is, given the realized $\mathbb{X}$, it is of no consequence whether or not it was generated by a random process; all that matters is how the number of bins with at least two observations scales with the total number of bins $k$.

Now suppose that $\var_{\theta} [Y | X = x] = \sigma^{2} (x)$ is heteroscedastic. Then to consider inference conditional on $\mathbb{X}$ we can take $\mathbb{P}_{n_{est}} \{\var_{\theta} [Y | X]\} = n^{-1} \sum_{i \in est} \var_{\theta} [Y | X_{i}]$ rather than $\psi (\theta) = \mathbb{E}_{\theta} [\var_{\theta} [Y | X]]$ as our object of inference, where we have again randomly split the data into an estimation and a training sample and we note if $X$ is random our object is random. In the random $X$ case, $\mathbb{P}_{n_{est}} \{\var_{\theta} [Y | X]\}$ is a $n^{1/2}$-consistent estimator of $\mathbb{E}_{\theta} [\var_{\theta} [Y | X]]$. It is natural to estimate $\mathbb{P}_{n_{est}} \{\var_{\theta} [Y | X]\}$ with either the empirical HOIF estimator $\hat{\psi}_{m, k}^{emp} = \hat{\psi}_{m, k}^{emp} (\widehat{\Omega}_{k}^{tr})$ defined earlier or with $\hat{\psi}_{m, k}^{emp} (\widehat{\Omega}_{k}^{est})$ where $\widehat{\Omega}_{k}^{est} = n^{-1} \sum_{i \in est} \bar{\mathsf{z}}_{k}(X_{i}) \bar{\mathsf{z}}_{k}(X_{i})^{\top}$ replaces $\widehat{\Omega}_{k}^{tr}$ in the estimator\footnote{$\widehat{\mathbb{IF}}_{22,k} (\widehat{\Omega}_{k}^{est})$ and $\hat{\psi}_{m, k}^{emp} (\widehat{\Omega}_{k}^{est})$ were introduced in Section S3 of the supplement of LMR and its properties were studied via simulation.}. Motivated by KBW's questions concerning conditional inference, it is of great interest to us to study the properties, conditional on $\mathbb{X}$, of $\hat{\psi}_{m, k}^{emp} \left(\widehat{\Omega }_{k}^{tr}\right)$ and  $\hat{\psi}_{m, k}^{emp} (\widehat{\Omega}_{k}^{est})$ as estimators of $\mathbb{P}_{n_{est}} \{\mathsf{var}_{\theta} [Y|X]\}$. 

 
\section{On universal inference}
\label{sec:universal} 
Rather than relying on asymptotic theory, \cite{wasserman2020universal} construct an ingenious universal confidence set by inverting concentration inequalities of the log-likelihood ratios, hence valid for any sample size $n$ for models with likelihood functions. This is another interesting contribution by two of the discussants to the large body of assumption-free/-lean inference literature (e.g. \citet{rinaldo2019bootstrapping}). In this section, we compare HOIF inference with universal inference for smooth nonlinear functionals. The take-home message is the following:
\begin{enumerate}
\item Universal inference requires a likelihood function. For undominated nonparametric problems, KBW consider (i) finite $k = k(n)$-dimensional sieves $\calM_{sub, k} = \left\{ \P_{\theta}; \theta \in \Theta_{sub, k} \subset \Theta \right\}$ of increasing dimension $k$, (ii) (under the conditions of Proposition 7 in \citet{wasserman2020universal}) an associated projection map $\theta \mapsto \tilde{\theta}_{\KL, k}$ from $\Theta$ onto $\Theta_{sub, k}$ that is the identity if $\theta \in \Theta_{sub, k}$ and otherwise is the minimizer in $\KL$-divergence between $\theta$ and $\Theta_{sub, k}$, and (iii) a preliminary estimator $\hat{\theta}_{\KL, k}$ for $\tilde{\theta}_{\KL, k}$ from a split sample. When a functional $\psi (\theta)$ is the object of inference, KBW propose to construct a universal confidence interval for the $\KL$-projection parameter $\tilde{\psi}_{\KL, k} (\theta) = \psi (\tilde{\theta}_{\KL} (\theta))$. Universal inference and HOIF inference may choose the same sieve. In that case, under certain laws $\P_{\theta}$, the $\KL$ projection $\tilde{\psi}_{\KL, k} (\theta)$ may be equal to the truncated parameter $\tilde{\psi}_{k} (\theta)$, as we will show below. Without further complexity reducing assumption to quantify the distance between the sieve and the true law $\P_{\theta}$, inference for $\tilde{\psi}_{\KL, k} (\theta)$ in universal inference or $\tilde{\psi}_{k} (\theta)$ in HOIF inference is the best one could hope for.

\item Universal confidence intervals are guaranteed to cover $\tilde{\psi}_{\KL, k} (\theta)$ at the nominal rate for any sample size $n$ under the conditions of Proposition 7 in \citet{wasserman2020universal}. However, we will show that the length of the confidence interval is of order $\Vert \tilde{\theta}_{\KL, k} - \hat{\theta}_{\KL, k} \Vert$\footnote{Here $\Vert \cdot \Vert$ is the $\ell_{2}$ norm of a $k$-dimensional vector.}, which is generically $\gg n^{-1 / 2}$ when $n > k \gg n^{1 / 2}$. In contrast, HOIF Wald confidence intervals for $\tilde{\psi}_{k} (\theta)$ centered at $\hat{\psi}_{m, k}^{emp}$ have length of order $n^{-1 / 2}$ because estimators $\hat{\psi}_{m, k}^{emp}$ of $\tilde{\psi}_{k} (\theta)$ typically have variance of order $1 / n$ for $k < n$ and bias less than $n^{-1/2}$ when $m$ is sufficiently large. When $\tilde{\psi}_{\KL, k} (\theta) = \tilde{\psi}_{k} (\theta)$, it follows that HOIF confidence intervals will be narrower than universal confidence intervals; however, nominal coverage of these HOIF Wald confidence intervals for $\tilde{\psi}_{k} (\theta)$ is guaranteed only in large samples. It is an interesting open problem to construct universal confidence intervals with guaranteed finite sample coverage of optimal length in large samples\footnote{First order influence function based intervals will enjoy the same properties as the higher order intervals when the test of the null hypothesis $\H_{0, k} (\delta)$ fails to reject in large samples.}.
\end{enumerate}

We now explain the above statements. To be concrete, we consider the following data generating process and functional of interest to simplify our analysis. We observe $n$ i.i.d. copies of $(A, X)$, with $X \sim \text{Uniform} ([0, 1])$ and $A \sim N(p(X), 1)$ with $p(x) \in L_{2} ([0, 1])$. The goal is to estimate $\psi (\theta) = \BE_{\theta} [p(X)^{2}]$. Let $\hat{p}(x)$ denote some initial machine learning estimator of the regression function $p(x)$, computed from an independent training sample treated as fixed. Again, we assume that the density of $X$ is known to focus on the important issues.

\subsection{Nonparametric universal inference}
\label{sec:kl}
We choose the following sieve for universal inference: $$\calM_{sub, k} = \left\{ N (p_{\theta_{k}} (X) = \hat{p}(X) + \theta_{k}^{\top} \zbar_{k}(X), 1); \theta_{k} \in \Theta_{sub, k} \right\}$$ for some $k \equiv k(n)$. The $\KL$-divergence between any member in this sieve $\calM_{sub, k}$ and the true law $A \sim N(p(X), 1)$ is
\begin{align*}
\mathsf{KL} (p, p_{\theta_{k}}) = \BE_{\theta} \left[ (p(X) - p_{\theta_{k}} (X))^{2} \right] = \BE_{\theta} \left[ (p(X) - \hat{p} (X) - \theta_{k}^{\top} \zbar_{k} (X) )^{2} \right].
\end{align*}
By definition, $\tilde{\theta}_{\KL, k}$ minimizes $\mathsf{KL} (p, p_{\theta_{k}})$ and hence $\tilde{\theta}_{\KL, k} = - \Omega_{k}^{-1} \BE_{\theta} \left[ \zbar_{k} (X) (\hat{p}(X) - p (X)) \right]$. By Proposition 7 in \cite{wasserman2020universal}, a nominal $1 - \alpha$ universal confidence set always covers $\tilde{\theta}_{\KL, k}$ with probability at least $1 - \alpha$.

Based on the sieve chosen above, $\tilde{\theta}_{\KL, k} = \tilde{\theta}_{k}$ (see Section \ref{sec:review}), and therefore $\tilde{\psi}_{\KL, k} (\theta) \equiv \int p_{\tilde{\theta}_{\KL, k}} (x)^{2} dx \equiv \int p_{\tilde{\theta}_{k}} (x)^{2} dx \equiv \tilde{\psi}_{k} (\theta)$. This is not surprising because the $\KL$-divergence between two normals is a quadratic form. Such isomorphism breaks down if $A \sim \mathsf{Bernoulli}(p(X))$ even with the same perturbation $p_{\theta_{k}}$. However, we can easily restore the isomorphism by replacing $\KL$-divergence with $\chi^{2}$-divergence. See Section \ref{app:kl} for more detail. Thus it will be interesting to generalize \citet[Proposition 7]{wasserman2020universal} from $\KL$ projection to projection based on general $\mathsf{f}$-divergences \citep{csiszar1964informationstheoretische, ali1966general}, which include $\KL$- and $\chi^{2}$-divergences as special cases.

\subsection{On the length of universal vs. HOIF confidence intervals}
\label{sec:length}
In this section, we suppose that $\tilde{\psi}_{\KL, k} (\theta)$ in universal inference and $\tilde{\psi}_{k} (\theta)$ in HOIF inference coincide. It is then natural to compare the length of the confidence intervals for $\tilde{\psi}_{\KL, k} (\theta)$ based on these two approaches.

Universal inference first estimates $\tilde{\theta}_{\KL, k}$ from half of the sample $\calD_{1}$ of size $n / 2$ by $\hat{\theta}_{\KL, k, \calD_{1}}$\footnote{$\hat{\theta}_{\KL, k, \calD_{1}} = - \ \widehat{\Omega}_{\mathcal{D}_{1}, k}^{-1} \BP_{n, \mathcal{D}_{1}} \left[ \zbar_{k} (X) (\hat{p} (X) - A) \right] \text{ with } \widehat{\Omega}_{\mathcal{D}_{1}, k}^{-1} \coloneqq \left\{ \BP_{n, \calD_{1}} \left[ \zbar_{k} (X_{i}) \zbar_{k} (X_{i})^{\top} \right] \right\}^{-1}$ where $\BP_{n, \calD}$ denotes the empirical measure over the sample $\calD$. }. Switching to another half of the sample $\mathcal{D}_{2}$ also of size $n / 2$, universal inference first finds a nominal $1 - \alpha$ confidence set $\widehat{\Theta}_{\mathcal{D}_{2}} (\alpha)$ (see equation \eqref{eq:set}) for $\tilde{\theta}_{\KL, k}$.

To construct a universal confidence interval for the functional $\tilde{\psi}_{\KL, k} (\theta)$ based on the confidence set $\widehat{\Theta}_{\mathcal{D}_{2}} (\alpha)$, \cite{wasserman2020universal} suggest to use the profile universal confidence interval $\widehat{\Psi}_{\calD_{2}}^{\mathsf{profile}} (\alpha)$ or the plug-in universal confidence interval $\widehat{\Psi}_{\calD_{2}}^{\mathsf{plug\text{-}in}} (\alpha)$ (see equations \eqref{eq:profile} or \eqref{eq:plugin}). In Section \ref{app:profile}, we show that both intervals have to contain the ``plug-in'' estimator $\int p_{\hat{\theta}_{\KL, k, \calD_{1}}} (x)^{2} d x$. Combined with the following lemma\footnote{For proof, see Theorem 1 in \citet{low1997nonparametric}.}, a lower bound is obtained on the expected length of both universal confidence intervals:
\begin{lem}
\label{prop:lower}
For any confidence interval $\Psi$ containing $\int p_{\hat{\theta}_{\KL, k, \calD_{1}}} (x)^{2} d x$, if it covers the target parameter $\tilde{\psi}_{\KL, k} (\theta)$ with probability at least $1 - \alpha$, then
\begin{align*}
\BE_{\theta} [L(\Psi) \vert \calD_{1}] & \geq (1 - \alpha) \left\vert \int p_{\hat{\theta}_{\KL, k, \calD_{1}}} (x)^{2} - p_{\tilde{\theta}_{\KL, k}} (x)^{2} d x \right\vert \\
& \equiv (1 - \alpha) \left\vert (\hat{\theta}_{\KL, k, \calD_{1}} - \theta_{\KL, \calD_{1}})^{\top} \Omega_{k} (\hat{\theta}_{\KL, k, \calD_{1}} + \theta_{\KL, \calD_{1}}) \right\vert.
\end{align*}
\end{lem}
Note that the lower bound given in Lemma \ref{prop:lower} is typically of the same order as $\Vert \hat{\theta}_{\KL, k, \calD_{1}} - \tilde{\theta}_{\KL, k} \Vert$. The length of a plug-in or a profile universal confidence interval is hence of order $\Vert \hat{\theta}_{\KL, k, \calD_{1}} - \tilde{\theta}_{\KL, k} \Vert$, which is typically of order $(k / n)^{1 / 2} \gg n^{-1/2}$ when $k \gg n^{1 / 2}$.

In contrast, HOIF estimators unbiasedly estimate $\tilde{\psi}_{k} (\theta) \equiv \tilde{\psi}_{\KL, k} (\theta)$ using a second order $U$-statistic $\hat{\psi}_{2, k}$\footnote{When $g$ is unknown, one needs to use empirical higher order $U$-statistic $\hat{\psi}_{m, k}^{emp}$ to further reduce the bias due to estimating $\Omega_{k} = \int \zbar_{k} (x) \zbar_{k} (x) g(x) dx$; see Section \ref{sec:review} and \citet{mukherjee2017semiparametric}.} with standard error of order $(1 / n)^{1/2} \vee (k / n^{2})^{1 / 2}$. Then a large sample HOIF Wald confidence interval $\hat{\psi}_{2, k} \pm z_{\alpha / 2} \widehat{\se} (\hat{\psi}_{2, k})$ typically has length of order $n^{-1/2}$ even if $k \gg n^{1/2}$ as long as $k < n$. Even if $k > n$, the length (of order $(k / n^{2})^{1 / 2}$) of an HOIF interval is still shorter than that of a universal confidence interval (of order $(k / n)^{1 / 2}$).

\begin{rem}
\label{rem:profile}
In \citet{murphy2000profile}, the authors (MvdV) showed that, under certain regularity conditions, the confidence interval for $\tilde{\psi}_{\KL, k} (\theta)$ based on inverting the profile likelihood ratio test of the hypothesis $\psi = \tilde{\psi}_{\KL, k} (\theta)$ attains nominal coverage in large samples and has length of order $n^{-1/2}$ even when $\tilde{\theta}_{\KL, k}$ cannot be estimated at rate $n^{-1/2}$. In contrast, as just shown, the universal confidence interval will not shrink at rate $n^{-1/2}$ in this setting; yet, if the universal confidence interval uses the MLE under the sieve model as a preliminary estimator of $\tilde{\theta}_{\KL, k}$, then the only difference between MvdV's interval and KBW's universal confidence interval is that for the former the unconditional MLE of $\tilde{\theta}_{\KL, k}$ in the numerator is computed from the same sample as the maximum profile likelihood in the denominator, while in the latter they come from different (split) samples. By examining the proofs in MvdV, one can see that the better rate depends crucially on an asymptotic expansion that exploits the fact that the numerator and denominator come from the same sample. 
\end{rem}

It is yet unclear to us how to reduce the length of a universal confidence interval. On the other hand, we have conjectured in LMR that a nonasymptotic HOIF confidence interval could be constructed by inverting exponential inequalities for $U$-statistics \citep{gine2000exponential, adamczak2006moment} but the theory is very challenging and doing so will necessarily increase the confidence interval's length. It will be interesting to investigate if such a non-asymptotic HOIF interval will still shrink faster than the universal confidence interval.

\begin{rem}[Final remark on regression and machine learning]
We agree with KBW that machine learning is more than prediction. In our paper, we equate ``machine learning'' with statistical prediction in order to connect with the most current use in causal inference. We believe and expect that many other aspects of machine learning, including clustering, density estimation with generative adversarial networks (GAN), dimension reduction/manifold learning, and optimal transport, will play more and more important roles in causal inference.
\end{rem}

\section*{Acknowledgement}
Lin Liu and James M. Robins were supported by the U.S. Office of Naval Research Grant N000141912446, National Institutes of Health (NIH) awards R01 AG057869 and R01 AI127271. Rajarshi Mukherjee’s research was partially supported by NSF grant EAGER-1941419.

\bibliographystyle{imsart-nameyear}
\bibliography{Master}

\begin{thebibliography}{30}

\bibitem[\protect\citeauthoryear{Adamczak}{2006}]{adamczak2006moment}
\begin{barticle}[author]
\bauthor{\bsnm{Adamczak},~\bfnm{Rados{\l}aw}\binits{R.}}
(\byear{2006}).
\btitle{Moment inequalities for U-statistics}.
\bjournal{The Annals of Probability}
\bvolume{34}
\bpages{2288--2314}.
\end{barticle}
\endbibitem

\bibitem[\protect\citeauthoryear{Ali and Silvey}{1966}]{ali1966general}
\begin{barticle}[author]
\bauthor{\bsnm{Ali},~\bfnm{Syed~Mumtaz}\binits{S.~M.}} \AND
  \bauthor{\bsnm{Silvey},~\bfnm{Samuel~D}\binits{S.~D.}}
(\byear{1966}).
\btitle{A general class of coefficients of divergence of one distribution from
  another}.
\bjournal{Journal of the Royal Statistical Society: Series B (Methodological)}
\bvolume{28}
\bpages{131--142}.
\end{barticle}
\endbibitem

\bibitem[\protect\citeauthoryear{Barron and
  Klusowski}{2018}]{barron2018approximation}
\begin{barticle}[author]
\bauthor{\bsnm{Barron},~\bfnm{Andrew~R}\binits{A.~R.}} \AND
  \bauthor{\bsnm{Klusowski},~\bfnm{Jason~M}\binits{J.~M.}}
(\byear{2018}).
\btitle{Approximation and estimation for high-dimensional deep learning
  networks}.
\bjournal{arXiv preprint arXiv:1809.03090}.
\end{barticle}
\endbibitem

\bibitem[\protect\citeauthoryear{Belloni et~al.}{2015}]{belloni2015some}
\begin{barticle}[author]
\bauthor{\bsnm{Belloni},~\bfnm{Alexandre}\binits{A.}},
  \bauthor{\bsnm{Chernozhukov},~\bfnm{Victor}\binits{V.}},
  \bauthor{\bsnm{Chetverikov},~\bfnm{Denis}\binits{D.}} \AND
  \bauthor{\bsnm{Kato},~\bfnm{Kengo}\binits{K.}}
(\byear{2015}).
\btitle{Some new asymptotic theory for least squares series: Pointwise and
  uniform results}.
\bjournal{Journal of Econometrics}
\bvolume{186}
\bpages{345--366}.
\end{barticle}
\endbibitem

\bibitem[\protect\citeauthoryear{Cai, Levine and Wang}{2009}]{cai2009variance}
\begin{barticle}[author]
\bauthor{\bsnm{Cai},~\bfnm{T~Tony}\binits{T.~T.}},
  \bauthor{\bsnm{Levine},~\bfnm{Michael}\binits{M.}} \AND
  \bauthor{\bsnm{Wang},~\bfnm{Lie}\binits{L.}}
(\byear{2009}).
\btitle{Variance function estimation in multivariate nonparametric regression
  with fixed design}.
\bjournal{Journal of Multivariate Analysis}
\bvolume{100}
\bpages{126--136}.
\end{barticle}
\endbibitem

\bibitem[\protect\citeauthoryear{Chen and Christensen}{2013}]{chen2013optimal}
\begin{barticle}[author]
\bauthor{\bsnm{Chen},~\bfnm{Xiaohong}\binits{X.}} \AND
  \bauthor{\bsnm{Christensen},~\bfnm{Timothy}\binits{T.}}
(\byear{2013}).
\btitle{Optimal uniform convergence rates for sieve nonparametric instrumental
  variables regression}.
\bjournal{arXiv preprint arXiv:1311.0412}.
\end{barticle}
\endbibitem

\bibitem[\protect\citeauthoryear{Chernozhukov, Newey and
  Singh}{2018}]{chernozhukov2018learning}
\begin{barticle}[author]
\bauthor{\bsnm{Chernozhukov},~\bfnm{Victor}\binits{V.}},
  \bauthor{\bsnm{Newey},~\bfnm{Whitney~K}\binits{W.~K.}} \AND
  \bauthor{\bsnm{Singh},~\bfnm{Rahul}\binits{R.}}
(\byear{2018}).
\btitle{Learning {L}2 Continuous Regression Functionals via Regularized {R}iesz
  Representers}.
\bjournal{arXiv preprint arXiv:1809.05224}.
\end{barticle}
\endbibitem

\bibitem[\protect\citeauthoryear{Csisz{\'a}r}{1964}]{csiszar1964informationstheoretische}
\begin{barticle}[author]
\bauthor{\bsnm{Csisz{\'a}r},~\bfnm{Imre}\binits{I.}}
(\byear{1964}).
\btitle{Eine informationstheoretische ungleichung und ihre anwendung auf beweis
  der ergodizitaet von markoffschen ketten}.
\bjournal{Magyer Tud. Akad. Mat. Kutato Int. Koezl.}
\bvolume{8}
\bpages{85--108}.
\end{barticle}
\endbibitem

\bibitem[\protect\citeauthoryear{Cui and Tchetgen~Tchetgen}{2019}]{cui2019bias}
\begin{barticle}[author]
\bauthor{\bsnm{Cui},~\bfnm{Yifan}\binits{Y.}} \AND
  \bauthor{\bsnm{Tchetgen~Tchetgen},~\bfnm{Eric}\binits{E.}}
(\byear{2019}).
\btitle{Selective machine learning of doubly robust functionals}.
\bjournal{arXiv preprint arXiv:1911.02029}.
\end{barticle}
\endbibitem

\bibitem[\protect\citeauthoryear{Gin{\'e}, Lata{\l}a and
  Zinn}{2000}]{gine2000exponential}
\begin{bincollection}[author]
\bauthor{\bsnm{Gin{\'e}},~\bfnm{Evarist}\binits{E.}},
  \bauthor{\bsnm{Lata{\l}a},~\bfnm{Rafa{\l}}\binits{R.}} \AND
  \bauthor{\bsnm{Zinn},~\bfnm{Joel}\binits{J.}}
(\byear{2000}).
\btitle{Exponential and moment inequalities for {U}-statistics}.
In \bbooktitle{High Dimensional Probability II}
\bpages{13--38}.
\bpublisher{Springer}.
\end{bincollection}
\endbibitem

\bibitem[\protect\citeauthoryear{Hayakawa and
  Suzuki}{2020}]{hayakawa2020minimax}
\begin{barticle}[author]
\bauthor{\bsnm{Hayakawa},~\bfnm{Satoshi}\binits{S.}} \AND
  \bauthor{\bsnm{Suzuki},~\bfnm{Taiji}\binits{T.}}
(\byear{2020}).
\btitle{On the minimax optimality and superiority of deep neural network
  learning over sparse parameter spaces}.
\bjournal{Neural Networks}
\bvolume{123}
\bpages{343--361}.
\end{barticle}
\endbibitem

\bibitem[\protect\citeauthoryear{Huang}{2003}]{huang2003local}
\begin{barticle}[author]
\bauthor{\bsnm{Huang},~\bfnm{Jianhua~Z}\binits{J.~Z.}}
(\byear{2003}).
\btitle{Local asymptotics for polynomial spline regression}.
\bjournal{The Annals of Statistics}
\bvolume{31}
\bpages{1600--1635}.
\end{barticle}
\endbibitem

\bibitem[\protect\citeauthoryear{Kennedy, Balakrishnan and
  Wasserman}{2020}]{kennedy2020discussion}
\begin{barticle}[author]
\bauthor{\bsnm{Kennedy},~\bfnm{Edward~H}\binits{E.~H.}},
  \bauthor{\bsnm{Balakrishnan},~\bfnm{Sivaraman}\binits{S.}} \AND
  \bauthor{\bsnm{Wasserman},~\bfnm{Larry~A}\binits{L.~A.}}
(\byear{2020}).
\btitle{Discussion of" On nearly assumption-free tests of nominal confidence
  interval coverage for causal parameters estimated by machine learning"}.
\bjournal{arXiv preprint arXiv:2006.09613}.
\end{barticle}
\endbibitem

\bibitem[\protect\citeauthoryear{Liu, Mukherjee and
  Robins}{2020}]{liu2020skepticism}
\begin{btechreport}[author]
\bauthor{\bsnm{Liu},~\bfnm{Lin}\binits{L.}},
  \bauthor{\bsnm{Mukherjee},~\bfnm{Rajarshi}\binits{R.}} \AND
  \bauthor{\bsnm{Robins},~\bfnm{James~M}\binits{J.~M.}}
(\byear{2020}).
\btitle{An assumption-lean skepticism test of inference validity for doubly
  robust functionals}
\btype{Technical Report},
\bpublisher{Available upon request}.
\end{btechreport}
\endbibitem

\bibitem[\protect\citeauthoryear{Liu, Mukherjee and Robins}{To
  appear}]{liu2020nearly}
\begin{barticle}[author]
\bauthor{\bsnm{Liu},~\bfnm{Lin}\binits{L.}},
  \bauthor{\bsnm{Mukherjee},~\bfnm{Rajarshi}\binits{R.}} \AND
  \bauthor{\bsnm{Robins},~\bfnm{James~M}\binits{J.~M.}}
(\byear{To appear}).
\btitle{On nearly assumption-free tests of nominal confidence interval coverage
  for causal parameters estimated by machine learning}.
\bjournal{Statistical Science}.
\end{barticle}
\endbibitem

\bibitem[\protect\citeauthoryear{Low}{1997}]{low1997nonparametric}
\begin{barticle}[author]
\bauthor{\bsnm{Low},~\bfnm{Mark~G}\binits{M.~G.}}
(\byear{1997}).
\btitle{On nonparametric confidence intervals}.
\bjournal{The Annals of Statistics}
\bvolume{25}
\bpages{2547--2554}.
\end{barticle}
\endbibitem

\bibitem[\protect\citeauthoryear{Mukherjee, Newey and
  Robins}{2017}]{mukherjee2017semiparametric}
\begin{barticle}[author]
\bauthor{\bsnm{Mukherjee},~\bfnm{Rajarshi}\binits{R.}},
  \bauthor{\bsnm{Newey},~\bfnm{Whitney~K}\binits{W.~K.}} \AND
  \bauthor{\bsnm{Robins},~\bfnm{James~M}\binits{J.~M.}}
(\byear{2017}).
\btitle{Semiparametric efficient empirical higher order influence function
  estimators}.
\bjournal{arXiv preprint arXiv:1705.07577}.
\end{barticle}
\endbibitem

\bibitem[\protect\citeauthoryear{Murphy and van~der
  Vaart}{2000}]{murphy2000profile}
\begin{barticle}[author]
\bauthor{\bsnm{Murphy},~\bfnm{Susan~A}\binits{S.~A.}} \AND
  \bauthor{\bparticle{van~der} \bsnm{Vaart},~\bfnm{Aad~W}\binits{A.~W.}}
(\byear{2000}).
\btitle{On profile likelihood}.
\bjournal{Journal of the American Statistical Association}
\bvolume{95}
\bpages{449--465}.
\end{barticle}
\endbibitem

\bibitem[\protect\citeauthoryear{Newey and Robins}{2018}]{newey2018cross}
\begin{barticle}[author]
\bauthor{\bsnm{Newey},~\bfnm{Whitney~K}\binits{W.~K.}} \AND
  \bauthor{\bsnm{Robins},~\bfnm{James~M}\binits{J.~M.}}
(\byear{2018}).
\btitle{Cross-fitting and fast remainder rates for semiparametric estimation}.
\bjournal{arXiv preprint arXiv:1801.09138}.
\end{barticle}
\endbibitem

\bibitem[\protect\citeauthoryear{Rinaldo, Wasserman and
  G'Sell}{2019}]{rinaldo2019bootstrapping}
\begin{barticle}[author]
\bauthor{\bsnm{Rinaldo},~\bfnm{Alessandro}\binits{A.}},
  \bauthor{\bsnm{Wasserman},~\bfnm{Larry}\binits{L.}} \AND
  \bauthor{\bsnm{G'Sell},~\bfnm{Max}\binits{M.}}
(\byear{2019}).
\btitle{Bootstrapping and sample splitting for high-dimensional,
  assumption-lean inference}.
\bjournal{The Annals of Statistics}
\bvolume{47}
\bpages{3438--3469}.
\end{barticle}
\endbibitem

\bibitem[\protect\citeauthoryear{Robins}{1988}]{robins1988confidence}
\begin{barticle}[author]
\bauthor{\bsnm{Robins},~\bfnm{James~M}\binits{J.~M.}}
(\byear{1988}).
\btitle{Confidence intervals for causal parameters}.
\bjournal{Statistics in Medicine}
\bvolume{7}
\bpages{773--785}.
\end{barticle}
\endbibitem

\bibitem[\protect\citeauthoryear{Robins et~al.}{2008}]{robins2008higher}
\begin{bincollection}[author]
\bauthor{\bsnm{Robins},~\bfnm{James}\binits{J.}},
  \bauthor{\bsnm{Li},~\bfnm{Lingling}\binits{L.}},
  \bauthor{\bsnm{Tchetgen~Tchetgen},~\bfnm{Eric}\binits{E.}} \AND
  \bauthor{\bparticle{van~der} \bsnm{Vaart},~\bfnm{Aad}\binits{A.}}
(\byear{2008}).
\btitle{Higher order influence functions and minimax estimation of nonlinear
  functionals}.
In \bbooktitle{Probability and Statistics: Essays in Honor of David A.
  Freedman}
\bpages{335--421}.
\bpublisher{Institute of Mathematical Statistics}.
\end{bincollection}
\endbibitem

\bibitem[\protect\citeauthoryear{Robins
  et~al.}{2009}]{robins2009semiparametric}
\begin{barticle}[author]
\bauthor{\bsnm{Robins},~\bfnm{James}\binits{J.}},
  \bauthor{\bsnm{Tchetgen~Tchetgen},~\bfnm{Eric}\binits{E.}},
  \bauthor{\bsnm{Li},~\bfnm{Lingling}\binits{L.}} \AND
  \bauthor{\bparticle{van~der} \bsnm{Vaart},~\bfnm{Aad}\binits{A.}}
(\byear{2009}).
\btitle{Semiparametric minimax rates}.
\bjournal{Electronic Journal of Statistics}
\bvolume{3}
\bpages{1305--1321}.
\end{barticle}
\endbibitem

\bibitem[\protect\citeauthoryear{Robins et~al.}{2016}]{robins2016technical}
\begin{barticle}[author]
\bauthor{\bsnm{Robins},~\bfnm{James}\binits{J.}},
  \bauthor{\bsnm{Li},~\bfnm{Lingling}\binits{L.}},
  \bauthor{\bsnm{Tchetgen~Tchetgen},~\bfnm{Eric}\binits{E.}} \AND
  \bauthor{\bparticle{van~der} \bsnm{Vaart},~\bfnm{Aad}\binits{A.}}
(\byear{2016}).
\btitle{Technical Report: Higher Order Influence Functions and Minimax
  Estimation of Nonlinear Functionals}.
\bjournal{arXiv preprint arXiv:1601.05820}.
\end{barticle}
\endbibitem

\bibitem[\protect\citeauthoryear{Robins et~al.}{2017}]{robins2017minimax}
\begin{barticle}[author]
\bauthor{\bsnm{Robins},~\bfnm{James~M}\binits{J.~M.}},
  \bauthor{\bsnm{Li},~\bfnm{Lingling}\binits{L.}},
  \bauthor{\bsnm{Mukherjee},~\bfnm{Rajarshi}\binits{R.}},
  \bauthor{\bsnm{Tchetgen~Tchetgen},~\bfnm{Eric}\binits{E.}} \AND
  \bauthor{\bparticle{van~der} \bsnm{Vaart},~\bfnm{Aad}\binits{A.}}
(\byear{2017}).
\btitle{Minimax estimation of a functional on a structured high-dimensional
  model}.
\bjournal{The Annals of Statistics}
\bvolume{45}
\bpages{1951--1987}.
\end{barticle}
\endbibitem

\bibitem[\protect\citeauthoryear{Rotnitzky, Smucler and
  Robins}{2019}]{rotnitzky2019characterization}
\begin{barticle}[author]
\bauthor{\bsnm{Rotnitzky},~\bfnm{Andrea}\binits{A.}},
  \bauthor{\bsnm{Smucler},~\bfnm{Ezequiel}\binits{E.}} \AND
  \bauthor{\bsnm{Robins},~\bfnm{James~M}\binits{J.~M.}}
(\byear{2019}).
\btitle{Characterization of parameters with a mixed bias property}.
\bjournal{arXiv preprint arXiv:1904.03725}.
\end{barticle}
\endbibitem

\bibitem[\protect\citeauthoryear{Schmidt-Hieber}{2020}]{schmidt2017nonparametric}
\begin{barticle}[author]
\bauthor{\bsnm{Schmidt-Hieber},~\bfnm{Johannes}\binits{J.}}
(\byear{2020}).
\btitle{Nonparametric regression using deep neural networks with {R}e{LU}
  activation function}.
\bjournal{To Appear in The Annals of Statistics}.
\end{barticle}
\endbibitem

\bibitem[\protect\citeauthoryear{Shen et~al.}{2019}]{shen2019optimal}
\begin{barticle}[author]
\bauthor{\bsnm{Shen},~\bfnm{Yandi}\binits{Y.}},
  \bauthor{\bsnm{Gao},~\bfnm{Chao}\binits{C.}},
  \bauthor{\bsnm{Witten},~\bfnm{Daniela}\binits{D.}} \AND
  \bauthor{\bsnm{Han},~\bfnm{Fang}\binits{F.}}
(\byear{2019}).
\btitle{Optimal estimation of variance in nonparametric regression with random
  design}.
\bjournal{arXiv preprint arXiv:1902.10822}.
\end{barticle}
\endbibitem

\bibitem[\protect\citeauthoryear{Wang et~al.}{2008}]{wang2008effect}
\begin{barticle}[author]
\bauthor{\bsnm{Wang},~\bfnm{Lie}\binits{L.}},
  \bauthor{\bsnm{Brown},~\bfnm{Lawrence~D}\binits{L.~D.}},
  \bauthor{\bsnm{Cai},~\bfnm{T~Tony}\binits{T.~T.}} \AND
  \bauthor{\bsnm{Levine},~\bfnm{Michael}\binits{M.}}
(\byear{2008}).
\btitle{Effect of mean on variance function estimation in nonparametric
  regression}.
\bjournal{The Annals of Statistics}
\bvolume{36}
\bpages{646--664}.
\end{barticle}
\endbibitem

\bibitem[\protect\citeauthoryear{Wasserman, Ramdas and
  Balakrishnan}{2020}]{wasserman2020universal}
\begin{barticle}[author]
\bauthor{\bsnm{Wasserman},~\bfnm{Larry}\binits{L.}},
  \bauthor{\bsnm{Ramdas},~\bfnm{Aaditya}\binits{A.}} \AND
  \bauthor{\bsnm{Balakrishnan},~\bfnm{Sivaraman}\binits{S.}}
(\byear{2020}).
\btitle{Universal inference}.
\bjournal{Proceedings of the National Academy of Sciences}
\bvolume{117}
\bpages{16880--16890}.
\end{barticle}
\endbibitem

\end{thebibliography}

\begin{frontmatter}

\title{Supplement to ``Rejoinder''}


\end{frontmatter}
\setcounter{section}{0}
\section{Tests when $\Omega_{\lowercase{k}}$ is unknown}
\label{app:emp}
In Section S3 of the online supplement of LMR and \citet[Section 4]{liu2020skepticism} we proposed the following test statistic 
\begin{equation} \label{eq:test}
\widehat{\chi}_{33, k} (\widehat{\Omega}_{k}^{-1}; z_{\alpha / 2}, \delta) = \mathbbm{1} \left\{ \frac{\vert \widehat{\IIFF}_{22 \rightarrow 33, k} (\widehat{\Omega}_{k}^{-1}) \vert}{\widehat{\se} [\hat{\psi}_{1}]} - z_{\alpha / 2} \frac{\widehat{\se} [\widehat{\IIFF}_{22 \rightarrow 33, k} (\widehat{\Omega}_{k}^{-1})]}{\widehat{\se} [\hat{\psi}_{1}]} > \delta \right\}
\end{equation}
which involves empirical HOIF estimators at order $m = 2, 3$. Here $\widehat{\se} [\widehat{\IIFF}_{22 \rightarrow 33, k} (\widehat{\Omega}_{k}^{-1})]$ is a consistent estimator of $\se_{\theta} [\widehat{\IIFF}_{22 \rightarrow 33, k} (\widehat{\Omega}_{k}^{-1})]$.

In order to show that the level and power properties of Theorem 3.2 and 4.2 of LMR also hold for $\widehat{\chi}_{33, k} (\widehat{\Omega}_{k}^{-1}; z_{\alpha / 2}, \delta)$, we need the following
\begin{equation}\label{cond}
\Vert \Pi [\hat{b} - b \vert \zbar_{k}] \Vert_{\infty} \leq C, \Vert \Pi [\hat{p} - p \vert \zbar_{k}] \Vert_{\infty} \leq C
\end{equation}
to hold with $\P_{\theta}$-probability 1 for some universal constant $C > 0$, when the residual functions $\hat{b} - b$ and $\hat{p} - p$ are only assumed to be bounded by some universal constant $C' > 0$.

We know the following basis functions that satisfy the above Condition \eqref{cond}: wavelets, B-spline and local polynomial partition series \citep{huang2003local, chen2013optimal, belloni2015some}, because they satisfy the extra condition in Condition W of LMR, but not in Condition \ref{cond:w}: $\Vert \Pi [\hat{b} - b \vert \zbar_{k}] \Vert_{\infty}$ and $\Vert \Pi [\hat{p} - p \vert \zbar_{k}] \Vert_{\infty}$ are $O(1)$ when $\hat{b} - b$ and $\hat{p} - p$ are $O(1)$. This is why we impose the additional assumption in Condition W compared to those in Condition \ref{cond:w}.

Without the additional assumption in Condition W, the level of the test $\widehat{\chi}_{33, k} (\widehat{\Omega}_{k}^{-1}; z_{\alpha / 2}, \delta)$ is still protected under the slightly weaker Condition \ref{cond:w}! In particular, we have:
\begin{proposition}\label{prop:level}
Under Condition \ref{cond:w}, $k \log^{2}(k) \ll n$ and $\Bias_{\theta, k} (\hat{\psi}_{1}) \neq o (\CSBias_{\theta, k} (\hat{\psi}_{1}))$, the test $\widehat{\chi}_{33, k} (\widehat{\Omega}_{k}^{-1}; z_{\alpha / 2}, \delta)$ is a valid asymptotic level $\alpha$ test of the surrogate null hypothesis $\H_{0, k} (\delta)$\footnote{The above result also appeared in \citet[Proposition 4.3]{liu2020skepticism}.}.
\end{proposition}

In terms of power, for the expected conditional covariance type functionals, in \citet{liu2020skepticism}, we showed that the standard error of an $m$-th order influence function estimator, when $m > 2$, has a term of the following order:
$$
\sqrt{\frac{k}{n}} \left( \BE_{\theta} \left[ \Pi [\hat{b} - b \vert \zbar_{k}] (X)^{2} \right] \right)^{1 / 2} \left( \BE_{\theta} \left[ \Pi [\hat{p} - p \vert \zbar_{k}] (X)^{2} \right] \right)^{1 / 2}
$$
which might exceed order of $1 / \sqrt{n}$. This term comes from the linear term in the Hoeffding decomposition of $m$-th order $U$-statistics and can dominate when $k < n$.

In terms of the power of $\widehat{\chi}_{33, k} (\widehat{\Omega}_{k}^{-1}; z_{\alpha / 2}, \delta)$ for surrogate null hypothesis $\H_{0, k} (\delta)$, for the expected conditional covariance, we need
\begin{equation}\label{eq:cond_cov}
\begin{split}
& \Bias_{\theta, k} (\hat{\psi}_{1})^{2} \equiv \left\{ \BE_{\theta} \left[ \Pi [\hat{b} - b \vert \zbar_{k}] (X) \Pi [\hat{p} - p \vert \zbar_{k}] (X) \right] \right\}^{2} \\
& \gg \frac{k}{n} \BE_{\theta} \left[ \Pi [\hat{b} - b \vert \zbar_{k}] (X)^{2} \right] \BE_{\theta} \left[ \Pi [\hat{p} - p \vert \zbar_{k}] (X)^{2} \right]
\end{split}
\end{equation}
to ensure that the rejection probability converges to 1 as $n \rightarrow \infty$. It says that $\Bias_{\theta, k} (\hat{\psi}_{1})$ should be greater in order than $\sqrt{\frac{k}{n}}$ fraction of its Cauchy-Schwarz upper bound for $\widehat{\chi}_{33, k} (\widehat{\Omega}_{k}^{-1}; z_{\alpha / 2}, \delta)$ to reject the surrogate null hypothesis $\H_{0, k} (\delta)$ with probability approaching 1.

For the expected conditional variance, the above requirement reduces to
\begin{equation}\label{eq:cond_var}
\BE_{\theta} \left[ \Pi [\hat{p} - p \vert \zbar_{k}] (X)^{2} \right] \gg \frac{k}{n}.
\end{equation}

In summary, to gather deeper understanding of the power of our test, a natural next step is to characterize the conditions on the (basis) functions $\zbar_{k}$ and the residual functions $\hat{b} - b$ and $\hat{p} - p$, under which Condition \eqref{eq:cond_cov} holds (or Condition \eqref{eq:cond_var} holds for the expected conditional variance).

\section{Technical detail for KBW's aggregation idea}
\label{app:order}
To study the mean and variance we first compute the expectation of $\widehat{\IIFF}_{22, \KBW} (\hat{f})$ given the selection sample\footnote{Note that we always condition on the training sample $\calD_{tr}$}:
\begin{align*}
\BE_{\theta} [\widehat{\IIFF}_{22, \KBW} (\hat{f}) \vert \calD_{sel}] = \left\{ \BE_{\theta} [(A - \hat{p}(X)) \hat{f}(X) \vert \calD_{sel}] \right\}^{2} \left\{ \BE_{\theta} [\hat{f}(X)^{2} \vert \calD_{sel}] \right\}^{-1} = \frac{\Bias_{\theta, k} (\hat{\psi}_{1})^{2} + \Delta_{\text{num}}}{\Bias_{\theta, k} (\hat{\psi}_{1}) + \Delta_{\text{denom}}}
\end{align*}
where
\begin{align*}
\Delta_{\text{num}} & \equiv \left\{ \frac{1}{n} \sum_{i \in \calD_{sel}} (A_{i} - \hat{p} (X_{i})) f(X_{i}) \right\}^{2} - \Bias_{\theta, k} (\hat{\psi}_{1})^{2} \\
\Delta_{\text{denom}} & \equiv \widehat{\IIFF}_{22, k} - \Bias_{\theta, k} (\hat{\psi}_{1}) + \frac{1}{n^{2}} \sum_{i \in \calD_{sel}} (A_{i} - \hat{p} (X_{i}))^{2} \zbar_{k} (X_{i})^{\top} \zbar_{k} (X_{i}).
\end{align*}
Also recall that $\hat{f} (x) = \zbar_{k} (x)^{\top} \hat{\beta}$ where $\hat{\beta} = \frac{1}{n} \Omega_{k}^{-1} \sum_{i \in \calD_{sel}} (A_{i} - \hat{p} (X_{i})) \zbar_{k} (X_{i})$.

We then prove the following:
\begin{lem}\label{lem:order}
Under Condition \ref{cond:w}, when $\Bias_{\theta} (\hat{\psi}_{1}) \gg n^{-1 / 2}$ and $k \gg k^{\ast}$, where $k^{\ast}$ is the minimax optimal choice of $k$ for estimating
\begin{align*}
\BE_{\theta} [\Bias_{\theta, k} (\hat{\psi}_{1})^{2} + \Delta_{\text{num}}] & \lesssim \Bias_{\theta, k} (\hat{\psi}_{1})^{2} + \frac{1}{n}, \var_{\theta} [\Bias_{\theta, k} (\hat{\psi}_{1})^{2} + \Delta_{\text{num}}] \lesssim \frac{1}{n} \left\{ \Bias_{\theta, k} (\hat{\psi}_{1})^{3} \vee \frac{\Bias_{\theta, k} (\hat{\psi}_{1})^{2}}{n} \right\}, \\
\BE_{\theta} [\Bias_{\theta, k} (\hat{\psi}_{1}) + \Delta_{\text{denom}}] & \lesssim \frac{k}{n}, \var_{\theta} [\Bias_{\theta, k} (\hat{\psi}_{1}) + \Delta_{\text{denom}}] \lesssim \frac{1}{n} \left\{ \Bias_{\theta, k} (\hat{\psi}_{1}) \vee \frac{k}{n} \right\}.
\end{align*}
\end{lem}

\begin{proof}
The order of $\BE_{\theta} [\Delta_{\text{num}}]$ and $\var_{\theta} [\Delta_{\text{num}}]$ are trivial and hence omitted. For $\Delta_{\text{denom}}$, $\widehat{\IIFF}_{22, k}$ is an unbiased estimator of $\Bias_{\theta, k} (\hat{\psi}_{1})$ and $\var_{\theta} [\widehat{\IIFF}_{22, k}] \lesssim \frac{1}{n} \left\{ \Bias_{\theta, k} (\hat{\psi}_{1}) \vee \frac{k}{n} \right\}$ by Theorem 2.6 of LMR. Finally, it is easy to see
\begin{align*}
\var_{\theta} \left[ \frac{1}{n^{2}} \sum_{i \in \calD_{sel}} (A_{i} - \hat{p} (X_{i}))^{2} \zbar_{k} (X_{i})^{\top} \Omega_{k}^{-1} \zbar_{k} (X_{i}) \right] \lesssim \frac{k}{n^{3}}.
\end{align*}
\end{proof}
Then $\var_{\theta} [\BE_{\theta} [\widehat{\IIFF}_{22, \KBW} (\hat{f}) \vert \calD_{sel}]]$ is typically of order $\frac{1}{n} \frac{\Bias_{\theta, k} (\hat{\psi}_{1})^{2}}{(k / n)^{2}}$, obtained by a Taylor expansion of the ratio of two random variables with standard errors dominated by their means, which is true when $\H_{0} (\delta)$ is false i.e. $\Bias_{\theta} (\hat{\psi}_{1}) \gg \frac{1}{\sqrt{n}}$ and $\Bias_{\theta, k} (\hat{\psi}_{1}) \asymp \Bias_{\theta} (\hat{\psi}_{1})$.

Next we need to compute $\BE_{\theta} [\var_{\theta} [\widehat{\IIFF}_{22, \KBW} (\hat{f}) \vert \calD_{sel}]]$. First
\begin{align*}
\var_{\theta} [\widehat{\IIFF}_{22, \KBW} (\hat{f}) \vert \calD_{sel}] \lesssim \frac{1}{n} \left( \frac{1}{n} \vee \frac{\left( \hat{\beta}^{\top} \Omega_{k} \beta \right)^{2}}{\hat{\beta}^{\top} \Omega_{k} \hat{\beta}} \right) \lesssim \frac{1}{n} \frac{\Bias_{\theta, k} (\hat{\psi}_{1})^{2} + \Delta_{\text{num}}}{\Bias_{\theta, k} (\hat{\psi}_{1}) + \Delta_{\text{denom}}}.
\end{align*}
Hence $\BE_{\theta} [\var_{\theta} [\widehat{\IIFF}_{22, \KBW} (\hat{f}) \vert \calD_{sel}]]$ is typically of order $\frac{\Bias_{\theta, k} (\hat{\psi}_{1})^{2}}{k}$ which is dominated by $\frac{1}{n} \frac{\Bias_{\theta, k} (\hat{\psi}_{1})^{2}}{(k / n)^{2}}$, the order of $\var_{\theta} [\BE_{\theta} [\widehat{\IIFF}_{22, \KBW} (\hat{f}) \vert \calD_{sel}]]$.

By the above heuristic arguments, $\widehat{\IIFF}_{22, \KBW} (\hat{f})$ has mean of order $\frac{\Bias_{\theta, k} (\hat{\psi}_{1})^{2}}{k / n}$, which dominates its standard error of order $\frac{\Bias_{\theta, k} (\hat{\psi}_{1}) / \sqrt{n}}{k / n}$ when $\Bias_{\theta, k} (\hat{\psi}_{1}) \gg n^{-1 / 2}$. 

\section{Technical details for universal inference}
\subsection{$\KL$- and $\chi^{2}$-divergences for $\mathsf{Bernoulli}$ model}
\label{app:kl}
Suppose that $A \in \mathsf{Bernoulli} (p(X))$ and the sieves are $A \sim \mathsf{Bernoulli} (p_{\theta_{k}} (X) = \hat{p}(X) + \theta_{k}^{\top} \zbar_{k}(X))$ for some $k \equiv k(n)$. Then
\begin{align*}
\mathsf{KL} (p, p_{\theta_{k}}) & = \BE_{\theta} \left[ \log \left( \frac{p (X)^{A} (1 - p(X))^{1 - A}}{p_{\theta_{k}} (X)^{A} (1 - p_{\theta_{k}} (X))^{1 - A}} \right) \right] \\
& = \BE_{\theta} \left[ p(X) \log \left( \frac{p(X)}{\hat{p} (X) + \theta_{k}^{\top} \zbar_{k} (X)} \right) + (1 - p(X)) \log \left( \frac{1 - p(X)}{1 - \hat{p} (X) - \theta_{k}^{\top} \zbar_{k} (X)} \right) \right] \\
& = - \ \BE_{\theta} \left[ p (X) \log \left( \hat{p} (X) + \theta_{k}^{\top} \zbar_{k} (X) \right) \right] - \BE_{\theta} \left[ (1 - p (X)) \log \left( 1 - \hat{p} (X) - \theta_{k}^{\top} \zbar_{k} (X) \right) \right] + C.
\end{align*}
Then taking derivative with respect to $\theta_{k}$,
\begin{align*}
0 = \frac{\partial \mathsf{KL} (p, p_{\theta_{k}})}{\partial \theta_{k}} = \BE_{\theta} \left[ \frac{\zbar_{k} (X) \zbar_{k} (X)^{\top} \theta_{k} + \zbar_{k} (X) (\hat{p}(X) - p(X))}{1 - (\hat{p} (X) + \theta_{k}^{\top} \zbar_{k} (X))^{2}} \right]
\end{align*}
has no closed form solution. However, if we consider $\chi^{2}$-divergence between $p$ and $p_{\theta_{k}}$ instead:
\begin{align*}
\chi^{2} (p, p_{\theta_{k}}) = \BE_{\theta} \left[ \left( p(X) - (\hat{p}(X) + \theta_{k}^{\top} \zbar_{k} (X)) \right)^{2} \right]
\end{align*}
which is then of the same form as the $\KL$-divergence when $A \sim N(p(X), 1)$. Hence under $\chi^{2}$-divergence, the isomorphism is restored. 

\subsection{The plug-in universal confidence interval is a subset of the profile universal confidence interval}
\label{app:profile}
\citet[Section 6]{wasserman2020universal} constructs confidence intervals for functional $\psi ((p_{\theta_{k}}, g))$ by the following profile likelihood of the functional $\psi ((p_{\theta_{k}}, g))$
\begin{equation}\label{eq:plugin}
\widehat{\Psi}_{\calD_{2}}^{\mathsf{profile}} (\alpha) = \left\{ \varphi: \frac{\mathcal{L}_{\calD_{2}} (\hat{\theta}_{\KL, k, \calD_{1}})}{\sup_{\theta_{k}: \psi (p_{\theta_{k}}) = \varphi} \mathcal{L}_{\calD_{2}} (\theta_{k})} \leq \frac{1}{\alpha} \right\}
\end{equation}
where $\mathcal{L}_{\calD_{2}} (\theta_{k})$ denotes the joint likelihood of the sample $\calD_{2}$ evaluated at the parameter $\theta_{k}$. They also stated without proof that $\widehat{\Psi}_{\calD_{2}}^{\mathsf{profile}} (\alpha)$ is equivalent to the following plug-in universal confidence interval:
\begin{equation}\label{eq:profile}
\widehat{\Psi}_{\calD_{2}}^{\mathsf{plug\text{-}in}} (\alpha) = \left\{ \varphi: \widehat{\Theta}_{\calD_{2}} (\alpha) \bigcap \psi^{-1} (\varphi) \neq \emptyset \right\}
\end{equation}
where
\begin{equation}
\label{eq:set}
\widehat{\Theta}_{\mathcal{D}_{2}} (\alpha) = \left\{ \theta_{k}: \BP_{n, \calD_{2}} \left\{ A - \hat{p} (X) - \theta_{k}^{\top} \zbar_{k} (X) \right\}^{2} \leq \BP_{n, \calD_{2}} \left\{ A - \hat{p} (X) - \hat{\theta}_{\KL, k, \mathcal{D}_{1}}^{\top} \zbar_{k} (X) \right\}^{2} + \frac{2}{n} \log \left( \frac{1}{\alpha} \right) \right\}.
\end{equation}
However, it should be noted that $\widehat{\Psi}_{\calD_{2}}^{\mathsf{profile}} (\alpha) \equiv \widehat{\Psi}_{\calD_{2}}^{\mathsf{plug\text{-}in}} (\alpha)$ under some regularity conditions, such as the compactness of the domain of $\psi$. Nonetheless, the following is always true.
\begin{lem}
\phantomsection
\label{lem:profile}
$\widehat{\Psi}_{\calD_{2}}^{\mathsf{plug\text{-}in}} (\alpha) \subseteq \widehat{\Psi}_{\calD_{2}}^{\mathsf{profile}} (\alpha)$.
\end{lem}
\begin{proof}
Choose any $\varphi \in \widehat{\Psi}_{\calD_{2}}^{\mathsf{plug\text{-}in}} (\alpha)$. By definition, for any $\theta_{k}' \in \psi^{-1} (\varphi)$, $\frac{\mathcal{L}_{\calD_{2}} (\hat{\theta}_{\KL, k, \calD_{1}})}{\mathcal{L}_{\calD_{2}} (\theta_{k}')} \leq \frac{1}{\alpha}$. Since $\frac{\mathcal{L}_{\calD_{2}} (\hat{\theta}_{\KL, k, \calD_{1}})}{\mathcal{L}_{\calD_{2}} (\theta_{k}')} \geq \frac{\mathcal{L}_{\calD_{2}} (\hat{\theta}_{\KL, k, \calD_{1}})}{\sup_{\theta_{k}: \psi (p_{\theta_{k}}) = \varphi} \mathcal{L}_{\calD_{2}} (\theta_{k})}$, $\varphi \in \widehat{\Psi}_{\calD_{2}}^{\mathsf{profile}} (\alpha)$. Hence $\widehat{\Psi}_{\calD_{2}}^{\mathsf{plug\text{-}in}} (\alpha) \subseteq \widehat{\Psi}_{\calD_{2}}^{\mathsf{profile}} (\alpha)$.
\end{proof}

In the example of Section \ref{sec:universal}, we have
$$
\widehat{\Psi}_{\calD_{2}}^{\mathsf{plug\text{-}in}} (\alpha) = \left\{ \int p_{\theta_{k}} (x)^{2} d x: \theta_{k} \in \widehat{\Theta}_{\mathcal{D}_{2}} (\alpha) \right\} \subseteq \widehat{\Psi}_{\calD_{2}}^{\mathsf{profile}} (\alpha).
$$
Since $\int p_{\hat{\theta}_{\KL, k, \calD_{1}}} (x)^{2} d x \in \widehat{\Psi}_{\calD_{2}}^{\mathsf{plug\text{-}in}} (\alpha)$, by Lemma \ref{lem:profile}, $\int p_{\hat{\theta}_{\KL, k, \calD_{1}}} (x)^{2} d x \in \widehat{\Psi}_{\calD_{2}}^{\mathsf{profile}} (\alpha)$.

\end{document}